%% file: arxiv.tex
\documentclass{jgaa-art-modified-for-arxiv}
\usepackage[T1]{fontenc}
\usepackage[utf8]{inputenc}
\usepackage{graphicx,amssymb,amsmath}
\usepackage{xspace}
\usepackage{hyperref}
\usepackage{enumerate}
\usepackage[subrefformat=parens,labelfont=default,font=small]{subcaption}
\usepackage{booktabs}
\usepackage{tikz,tikz-3dplot}
\usetikzlibrary{calc,intersections}
\usepackage{cite}

\usepackage[misc,geometry]{ifsym}

\newtheorem{theorem}{Theorem}[section]
\newtheorem{lemma}[theorem]{Lemma}
\newtheorem{open}[theorem]{Open Problem}

\newtheorem{observation}[theorem]{Observation}
\newtheorem{proposition}[theorem]{Proposition}
\newtheorem{corollary}[theorem]{Corollary}

\newcommand{\Gmed}{\ensuremath{G_\mathrm{med}}\xspace}
\newcommand{\K}{\ensuremath{\mathcal{K}}}
\newcommand{\xydist}[3]{d_{#1}(#2,#3)}

\graphicspath{{figures/}}

\definecolor{darkgreen}{rgb}{0,0.53,0.24}
\definecolor{plum}{RGB}{221,160,221}
\definecolor{lgray}{RGB}{152, 175, 199}

\definecolor{cyan}{RGB}{0,224,224}
\definecolor{lbl}{RGB}{48, 144, 199}
\definecolor{yl}{RGB}{255,204,0}
\definecolor{sr}{rgb}{0.55,0.2,0.25}

\definecolor{ci}{RGB}{46,189,78}
\definecolor{co}{RGB}{11,83,69}
\definecolor{lbl}{RGB}{130,202,255}
\definecolor{lgray}{RGB}{152, 175, 199}

\begin{document}


\doi{}
\Issue{0}{0}{0}{0}{0} 
\HeadingAuthor{Evans et al.} 
\HeadingTitle{Representing Graphs and Hypergraphs  by Polygons in 3D}
\title{Representing Graphs and Hypergraphs by~Touching Polygons in 3D}

\Ack{A preliminary version of this work appeared in Proc.\ 27th Int.\
  Symp.\ Graph Drawing \& Network Vis.\ (GD 2019) 
  \cite{erssw-rghtp-gd2019}.  W.E.~and N.S.~were funded by an NSERC
  Discovery grant and in part by the Institute for Computing,
  Information and Cognitive Systems (ICICS) at UBC.  P.Rz.~was
  supported by the ERC starting grant CUTACOMBS (no.~714704).
  A.W.~was funded by the German Research Foundation (DFG) under grant
  406987503 (WO 758/10-1).  \mbox{C.-S.Sh.} was supported by the
  National Research Foundation of Korea (NRF) grant funded by the
  Korea government (MSIT) (no.~2019R1F1A1058963).}

\author[first]{William Evans}{will@cs.ubc.ca}
\author[second,third]{Pawe\l~Rz\k{a}\.zewski}{p.rzazewski@mini.pw.edu.pl}
\author[first]{Noushin~Saeedi}{noushins@cs.ubc.ca}
\author[forth]{Chan-Su~Shin}{cssin@hufs.ac.kr}
\author[fifth]{Alexander Wolff}{https://orcid.org/0000-0001-5872-718X}

\affiliation[first]{University of British Columbia, Vancouver, Canada}
\affiliation[second]{Warsaw University of Technology, \\ Faculty of
  Mathematics and Information Science, Warszawa, Poland}
\affiliation[third]{Institute of Informatics, University of Warsaw,
  Warszawa, Poland}
\affiliation[forth]{Hankuk University of Foreign Studies, Yongin,
  Republic of Korea}
\affiliation[fifth]{Universit\"at W\"urzburg, W\"urzburg, Germany}


\submitted{March 2020}%
\reviewed{}%
\revised{}%
\reviewed{}%
\revised{}%
\accepted{}%
\final{}%
\published{}%
\type{Regular paper}%
\editor{}%

\maketitle

\vspace*{-2ex}

\centerline{\textit{Dedicated to Honza Kratochv\'il on his 60th birthday.}}

\vspace*{2ex}


\begin{abstract}
  Contact representations of graphs have a long history.
  Most research has focused on problems in 2D, but 3D contact
  representations have also been investigated, mostly concerning
  fully-dimensional geometric objects such as spheres or cubes.  In
  this paper we study contact representations with convex polygons
  in~3D.  We show that every graph admits such a representation.
  Since our representations use super-polynomial coordinates, we also
  construct representations on grids of polynomial size for specific
  graph classes (bipartite, subcubic).  For hypergraphs, we represent
  their duals, that is, each vertex is represented by a point and each
  edge by a polygon.  We show that even regular and quite small
  hypergraphs do not admit such representations.  On the other hand,
  the two smallest Steiner triple systems can be represented.
\end{abstract}

\Body 

\section{Introduction}

Representing graphs as the contact of geometric objects has been an
area of active research for many years (see Hlin\v{e}n\'{y} and
Kratochv\'{i}l's survey~\cite{HK01} and Alam's thesis~\cite{AlamThesis15}).
Most of this work concerns representations in 2D, though
there has been some interest in three-dimensional representation as well
\cite{Thomassen86,befhkllrw-rgtc-GD12,aekptu-crg3d-WADS15,FF11,akk-cgcpb-SOFSEM16}.
Representations in 3D typically use 3D geometric objects
that touch properly, i.e., their intersection is a positive area 2D
face. 
In contrast, our main focus is on contact representation of graphs and
hypergraphs using non-intersecting (open, ``filled'') planar polygons
in 3D.
Two polygons are in \emph{contact} if they share a corner point.
Note that two triangles that share two corner points do not intersect and
a triangle and rectangle that share two corners, even diagonally
opposite ones, also do not intersect.
However, no polygon contains a corner of another except at its own corner.
A \emph{contact representation of a graph in 3D} is a set of
non-intersecting polygons in 3D that represent vertices. Two polygons share a 
corner point if and only if they represent adjacent vertices and
each corner point corresponds to a distinct edge.
We can see a contact representation of a graph $G = (V,E)$ as a certain drawing of its 
\emph{dual hypergraph} $H_G=(E,\{E(v)\mid v \in V\})$ which has a vertex for every edge
of~$G$, and a hyperedge for every vertex~$v$ of~$G$, namely the
set~$E(v)$ of edges incident to~$v$. 
We extend this idea to arbitrary hypergraphs: A \emph{non-crossing drawing of a hypergraph in 3D}
is a set of non-intersecting polygons in 3D that represent \emph{edges}. Two
polygons share a corner point if and only if they represent edges that
contain the same vertex and each corner point corresponds to a
distinct vertex.
It is straightforward to observe that the set of contact representations of a graph $G$ is the same as
the set of non-crossing drawings of  $H_G$.

Many people have studied ways to represent hypergraphs
geometrically~\cite{jp-hpcdvd-JGT87,bkmsv-psh-JGAA11,bcps-pbsh-JDA12},
perhaps starting with Zykov~\cite{z-h-UMN74}.  A natural motivation of
this line of research was to find a nice way to represent
combinatorial configurations~\cite{g-dc-GD95} such as Steiner systems
(for an example, see Fig.~\ref{fig:fano}).  The main focus in
representing hypergraphs, however, was on drawings in the plane.
By using polygons to represent hyperedges in 3D, we gain some
additional flexibility though still not all hypergraphs can be
realized.  
Our work is related to Carmesin's work~\cite{Carmesin19} on
a Kuratowski-type characterization of 2D simplicial complexes (sets
composed of points, line segments, and triangles)
that have an embedding in 3-space.
Our representations are sets of planar polygons (not just triangles)
that arise from hypergraphs.
Thus they are less expressive than Carmesin's topological 2D
simplicial complexes and are more restricted.
In particular, if two hyperedges share three vertices, the hyperedges
must be coplanar in our representation.

Our work is also related to that of Ossona de
Mendez~\cite{o-rp-JGAA02}.  He showed that a hypergraph whose
vertex--hyperedge inclusion order has poset dimension~$d$ can be
embedded into~$\mathbb{R}^{d-1}$ such that every vertex corresponds to
a unique point in~$\mathbb{R}^{d-1}$ and every hyperedge corresponds
to the convex hull of its vertices.  The embedding ensures that the
image of a hyperedge does not contain the image of a vertex and, for
any two hyperedges~$e$ and~$e'$, the convex hulls of $e\setminus e'$
and of $e' \setminus e$ don't intersect.  In particular, the images of
disjoint hyperedges are disjoint.  Note that both Ossona de Mendez and
we use triangles to represent hyperedges of size~3, but for larger
hyperedges, he uses higher-dimensional convex subspaces.
Note also that the method of Ossona de Mendez may insist on a higher
dimension than actually needed.  For example, every graph (seen as a
2-uniform hypergraph) can be drawn with non-intersecting 
straight-line segments in 3D, but the
vertex--hyperedge inclusion order of~$K_{13}$ has poset
dimension~5~\cite{HOSTEN99}, so the method of Ossona de Mendez
needs~4D for a straight-line drawing of~$K_{13}$.

\paragraph{Our contribution.}
All of our representations in this paper use convex polygons while our
proofs of non-representability hold even permitting non-convex polygons.
We first show that recognizing segment graphs in 3D is
$\exists\mathbb{R}$-complete.

We show that every graph on $n$ vertices with minimum vertex-degree 3
has a contact representation by convex polygons in 3D, though the
volume of the drawing using integer coordinates is at least
exponential in~$n$; see Section~\ref{sec:graphs}.

For some graph classes, we give 3D drawing algorithms which require
polynomial volume.  Table~\ref{tab:gv} summarizes our results.
When we specify the
volume of the drawing, we take the product of the number of grid lines in
each dimension (rather than the volume of a bounding box), so that a
drawing in the xy-plane has non-zero volume.
Some graphs, such as the squares of even cycles, have
particularly nice representations using only unit squares; see
Fig.~\ref{fig:c2}\subref{fig:c2_67}.

\begin{table}[tb]
  \centering
  \caption{Required grid volume and running times of our algorithms for
    drawing $n$-vertex graphs of certain graph classes in 3D}
  \label{tab:gv}
  \smallskip
  \begin{tabular}{l@{\qquad}c@{\quad}c@{\quad}c@{\quad}c@{\quad}c}
    \toprule
    Graph & general & bipartite & 1-plane & 2-edge-conn. & subcubic \\
    class &         &	      &	cubic   & cubic & \\
    \midrule
    Volume    & super-poly & $O(n^4)$ & $O(n^2)$ & $O(n^2)$ & $O(n^3)$ \\
    Runtime   & $O(n^2)$ & linear & linear & $O(n \log^2 n)$ & $O(n \log^2 n)$ \\
    Reference & Thm.~\ref{thm:general} & Thm.~\ref{thm:bipartite}
    & Thm.~\ref{thm:1-planar} & Lem.~\ref{lem:two-connected-cubic} & Thm.~\ref{thm:cubic} \\
    \bottomrule
  \end{tabular}
\end{table}
 
For hypergraphs our results are more preliminary.
There are examples as simple as the hypergraph on six
vertices with all triples of vertices as hyperedges that cannot be
drawn using non-intersecting triangles; see Section~\ref{sec:hypergraphs}.
We show that
hypergraphs with too many edges of cardinality~4 such as
Steiner quadruple systems do not admit non-crossing drawings using convex
quadrilaterals (they in fact do not admit non-crossing drawings with any quadrilaterals if the number of vertices is sufficiently large).  On the other hand, we show that the two smallest
Steiner triple systems can be drawn using triangles.  (We define these
two classes of hypergraphs in Section~\ref{sec:hypergraphs}.)

\section{Graphs}
\label{sec:graphs}

It is easy to draw graphs in 3D using points as vertices and
non-crossing line segments as edges -- any set of points in general
position (no three collinear and no four coplanar) will support any set
of edge segments without crossings.
A more difficult problem is to represent a graph in 3D using polygons
as vertices where two polygons intersect to indicate an edge (note
that here we do not insist on a {\emph contact} representation, i.e.,
polygons are allowed to intersect arbitrarily).
Intersection graphs of convex polygons in 2D have been studied
extensively \cite{ll-cpig-GD10}.
Recognition is
$\exists\mathbb{R}$-complete~\cite{s-cgtp-09} (and thus in PSPACE
since $\exists \mathbb{R} \subseteq$
PSPACE~\cite{DBLP:conf/stoc/Canny88})
even for segments (polygons with only two vertices).

Every complete graph trivially admits an intersection representation
by line segments in~2D.  Not every graph, however, can be represented
in this way, see e.g., Kratochv\'il and
Matou\v{s}ek~\cite{DBLP:journals/jct/KratochvilM94}.
Moreover, they show that recognizing
intersection graphs of line segments in the plane, called \emph{segment graphs}, is
$\exists\mathbb{R}$-complete.
It turns out that a similar hardness result holds for recognizing
intersection graphs of straight-line segments in 3D (and actually in
any dimension).
The proof modifies the corresponding proof for~2D by
Schaefer~\cite{s-cgtp-09}. See also the excellent exposition of the
proof by Matou\v{s}ek~\cite{DBLP:journals/corr/Matousek14}.

\begin{theorem}
  \label{thm:recognition}
  Recognizing segment graphs in 3D is $\exists\mathbb{R}$-complete.
\end{theorem}

\begin{proof}
  Clearly the problem is in $\exists\mathbb{R}$, so we immediately turn
  to hardness.  The proof is a reduction from \textsc{Stretchability},
  where we are given a combinatorial description of a collection of
  pseudolines, and we ask whether there is a collection of straight
  lines with the same description. 

  We start with a brief description of the original reduction, in
  the 2-dimensional
  case~\cite{s-cgtp-09,DBLP:journals/corr/Matousek14}. Following
  Schaefer and Matou\v{s}ek, we will describe the construction geometrically.
  {This is a convenient way to describe how to obtain a
  graph $G$ from the combinatorial description of the collection of
  pseudolines so that $G$ is a segment intersection graph if and
  only if the collection of pseudolines is stretchable.}
  More specifically, we will assume that the input combinatorial description
  can be arranged by straight lines, and will describe a corresponding 
  arrangement of straight-line segments, which forms an intersection representation
  of the constructed graph $G$. Formally, the input of the recognition problem is the
  purely combinatorial description of the graph $G$, not the representation.
  The construction ensures that if $G$ is a segment
  intersection graph, then every intersection representation by segments must
  be equivalent to the intended one.
  
  In his reduction,
  Schaeffer~\cite{s-cgtp-09,DBLP:journals/corr/Matousek14}
  constructs an arrangement of segments with the desired combinatorial
  description.  We call the segments in this arrangement \emph{original
    segments}. He introduces three new, pairwise intersecting segments
  $a$, $b$, and $c$, called \emph{frame segments}. They are placed in
  such a way that every original segment intersects at least two frame
  segments, and all intersections of original segments take place
  inside the triangle bounded by $a$, $b$, and $c$;
  see Fig.~\ref{fig:frame}.
  Next, for every original segment, he adds many new segments, called
  \emph{order segments}. Their purpose is to ensure that every
  representation of the constructed graph $G$ with intersecting
  segments has the desired ordering of crossings of original segments;
  see~Fig.~\ref{fig:coplanar}~(left).

 \begin{figure}[h]
    \centering
    \includegraphics[scale=1,page=1]{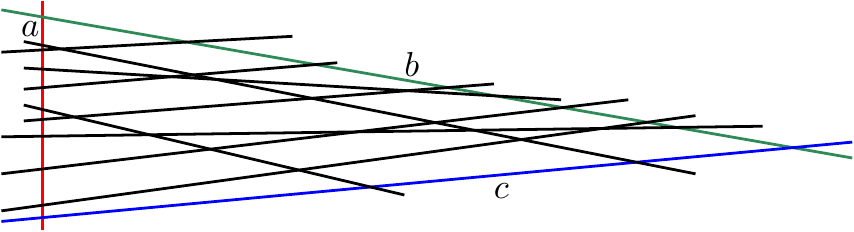}
    \caption{Original segments and frame segments.}
    \label{fig:frame}
  \end{figure}
 
  \begin{figure}[tb]
      \includegraphics[page=2]{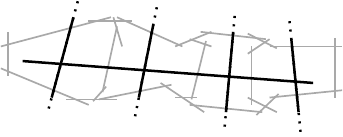}
      \hfill
      \includegraphics[page=3]{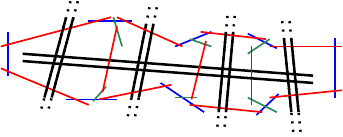} 
      \caption{Left: Placement of order segments (thin lines). Original
        segments and frame segments are drawn with thick lines. Right:
        Twins force all segments to be coplanar. Each segment
        drawn red intersects two original or twin segments. Each segment drawn blue intersects two red
        segments. Finally, each green segment intersects a blue and a
        red segment.
        }
      \label{fig:coplanar}
  \end{figure}

  In order to show recognition hardness in 3D, we introduce some new
  segments (new vertices to $G$), obtaining a new graph
  $G'$.  For each original segment $s$, we introduce its \emph{twin}
  $s'$, i.e., a parallel non-overlapping segment with exactly the same
  neighbors as $s$. This completes the construction of $G'$.

  Now we argue that in every representation of $G'$, all segments
  from the representation are coplanar.  First, note that the frame
  segments define a plane, let us call it the \emph{base
    plane}. Moreover, recall that each original segment intersects at
  least two frame segments, so it also lies in the base plane. By the
  same argument, also twins of original segments lie in the base
  plane. Next, note that each order segment that intersects an
  original segment of $G$ now intersects an original segment and its
  twin, which forces it to lie in the base plane. It is
  straightforward to verify that all other order segments are forced
  to lie in the base plane too; see Fig.~\ref{fig:coplanar}~(right).

  It is easy to verify (see, e.g.,
  \cite{DBLP:journals/jgaa/CardinalFMTV18} for a similar argument) that
  $G'$ can be represented by intersecting segments in 3D if and only if
  $G'$ (and also $G$) can be represented by intersecting segments in 2D, and consequently,
  if and only if the initial instance of \textsc{Stretchability} is a
  yes-instance.
\end{proof}

We consider \emph{contact representations} of graphs in 3D where no
polygons are allowed to intersect except at their corners,
and two polygons share a corner if and only if they represent adjacent vertices.
We start by describing how to construct a contact representation for
any graph using convex polygons, which requires at least exponential
volume, and then describe constructions for graph families that use only
polynomial volume.

\subsection{General Graphs}
\label{sub:general}

\begin{lemma}
  \label{lem:grid}
  For every positive integer $n \ge 3$, there exists an arrangement of
  $n$ lines $\ell_1,\ell_2,\ldots,\ell_n$
  with the following two properties:
  \begin{enumerate}[({A}1)]
  \item \label{itm:intersect} line $\ell_i$ intersects lines $\ell_1, \ell_2, \dots, \ell_{i-1}, \ell_{i+1}, \dots \ell_n$ in
    this order, and
  \item \label{itm:decrease} distances between the intersection points on line~$\ell_i$ decrease exponentially, i.e., for every $i$ it holds that
  \begin{align}
  \xydist{i}{j+2}{j+1} &\le \xydist{i}{j+1}{j}/2 
  && \text{for } j \in \{1,\dots,i-3\}\label{ineq:small}\\
  \xydist{i}{i+1}{i-1} &\le \xydist{i}{i-1}{i-2}/2\label{ineq:minus1}\\
  \xydist{i}{i+2}{i+1} &\le \xydist{i}{i+1}{i-1}/2 \\
  \xydist{i}{j+2}{j+1} &\le \xydist{i}{j+1}{j}/2 
  && \text{for } j \in \{i+1,\dots,n-2\},
  \end{align}
  where $\xydist{i}{j}{k}$ is the xy-plane distance between $p_{i,j}$
  and $p_{i,k}$ and $p_{i,j} = p_{j,i}$ is the intersection point of
  $\ell_i$ and $\ell_j$.
  \end{enumerate}
\end{lemma}

\begin{proof}
  We construct the grid incrementally.  We start with the x-axis as
  $\ell_1$, the y-axis as $\ell_2$, and the line
  through $(1,0)$ and $(0,-1)$ as $\ell_3$; see Fig.~\ref{fig:grid}.
  Now suppose that $i > 3$, we have constructed lines $\ell_1,\ell_2,\ldots,\ell_{i-1}$,
  and we want to construct~$\ell_i$.  
  We fix $p_{i-1,i}$ to satisfy $\xydist{i-1}{i}{i-2} =
  \xydist{i-1}{i-2}{i-3}/2$ then rotate a copy of line~$\ell_{i-1}$
  clockwise around $p_{i-1,i}$ until it (as $\ell_i$) satisfies
  another of the inequalities in
  \eqref{ineq:small} with equality.
  Note that during this rotation, all inequalities in~(A\ref{itm:decrease}) 
  are satisfied and we do not move any previously constructed lines,
  so the claim of the lemma follows.
\end{proof}

\begin{figure}[tb]
  \begin{minipage}[b]{.44\linewidth}
    \centering
    \includegraphics{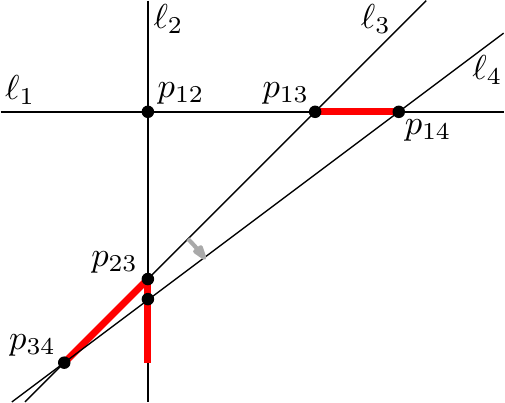}
    \caption{Construction of $\ell_4$ in the proof of Lemma~\ref{lem:grid}.}
    \label{fig:grid}
  \end{minipage}
  \hfill
  \begin{minipage}[b]{.51\linewidth}
    \centering
    \includegraphics{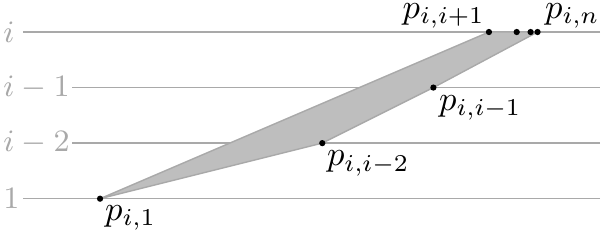}
    \caption{The polygon~$P_i$ that represents vertex~$i$ of~$K_n$.}
    \label{fig:polygon-for-Kn}
  \end{minipage}
\end{figure}

\begin{theorem}
  \label{thm:general}
  For every $n \ge 3$, the complete graph $K_n$ admits a contact representation by
  non-degenerate convex polygons in~3D, each with at most $n-1$
  vertices.  Such a representation can be computed in $O(n^2)$ time
  (assuming unit cost for arithmetic operations on coordinates).
\end{theorem}

\begin{proof}
  Take a grid according to Lemma~\ref{lem:grid} (thus, each line $\ell_i$ is in the xy-plane). 
  We lift each intersection point $p_{i,j}$ so that the z-coordinate
  of $p_{i,j}$ becomes $\min \{i,j\}$. We represent vertex~$i$ by
  polygon~$P_i$, which we define to be the convex hull of 
  $\{p_{i,1}, p_{i,2}, \dots, p_{i,i-1}, p_{i,i+1}, \dots, p_{i,n}\}$.
  Note that~$P_i$ is contained in the vertical plane that
  contains line~$\ell_i$; see Fig.~\ref{fig:polygon-for-Kn}.
  To avoid that~$P_1$ is degenerate, we reduce the
  z-coordinate of $p_{1,2}$ slightly.

We claim that the counterclockwise order of vertices around~$P_i$, for
$i=2,\dots,n-1$ is
\[p_{i,1}, p_{i,2}, \dots, p_{i,i-1}, p_{i,n}, p_{i,n-1}, \dots,
  p_{i,i+1}, p_{i,1}.\]
Similarly, we claim that the counterclockwise order of vertices
around~$P_1$ is
$p_{1,2},p_{1,n}, \dots, p_{1,3}, p_{1,2}$,
and the order around~$P_n$ is
$p_{n,1}, p_{n,2}, \dots, p_{n,n-1}, p_{n,1}$.
Note that a polygon with such an ordering is simple (i.e., it does not self-intersect).
We prove these claims by showing that the angle formed by any three
consecutive points in these orders is bounded by~$\pi$.
We can easily verify this for~$P_1$ and~$P_n$. In the following we
assume that $i \in \{2,\dots,n-1\}$.
  Clearly the angles $\angle p_{i,i+1} p_{i,1} p_{i,2}$ and $\angle p_{i,i-1} p_{i,n} p_{i,n-1}$ are at most~$\pi$.
  For $j=2,\dots,i-2$, we have
  $\angle p_{i,j-1}p_{i,j}p_{i,j+1}<\pi$, which is due to the fact that
  the z-coordinates increase in each step by~1, while the distances
  decrease (property~(A\ref{itm:decrease})).  Note that $\angle
  p_{i,i+1},p_{i,i+2},p_{i,i+3}= \dots = \angle
  p_{i,n-2},p_{i,n-1},p_{i,n} = \pi$.  Finally, we claim that
  $\angle p_{i,i-2},p_{i,i-1},p_{i,n} < \pi$.  Clearly,
  $z(p_{i,i-1})-z(p_{i,i-2})=1=z(p_{i,n})-z(p_{i,i-1})$, where $z(p)$ denotes the z-coordinate of point $p$.  The claim
  follows by observing that, due to property~(A\ref{itm:decrease}) and
  the geometric series formed by the distances,
  \[\xydist{i}{i-1}{n} = \xydist{i}{i-1}{i+1} + 
  \sum_{k=i+1}^{n-1} \xydist{i}{k}{k+1} <
  2\xydist{i}{i-1}{i+1} \le \xydist{i}{i-2}{i-1}.\] 
  It remains to show that, for $1 \le i < j \le n$, polygons~$P_i$
  and~$P_j$ do not intersect other than in~$p_{i,j}$.  This is simply
  due to the fact that~$P_j$ is above~$P_i$ in $p_{i,j}$, and
  lines~$\ell_i$ and~$\ell_j$ only intersect in (the projection of) this point.
\end{proof}

\begin{corollary}
  \label{cor:general}
  Every graph with minimum vertex-degree~3 admits a contact
  representation by convex polygons in~3D.
\end{corollary}

\begin{proof}
  Let $n$ be the number of vertices of the given graph~$G=(V,E)$.  We
  use the contact representation of~$K_n$ and modify it as follows.
  For every pair $\{i,j\} \not\in E$, just remove the point~$p_{i,j}$
  before defining the convex hulls.

  In the above construction for $K_n$, we had to make sure that 
  polygon~$P_1$ is not degenerate.  For general graphs, we have the
  same problem for any polygon~$P_i$ with the property that every
  polygon~$P_j$ that is adjacent to~$P_i$ has index $j>i$.  Let $k$
  be the smallest index such that $P_k$ is incident to~$P_i$.  
  Now if $k > i$ then we slightly reduce the z-coordinate of~$p_{i,k}$.
  (In the construction for $K_n$, we did this only to~$p_{1,2}$.)
\end{proof}

We can make the convex polygons of our construction strictly convex if
we slightly change the z-coordinates.
For example, decrease the z-coordinate of $p_{i,j}$ by
$\delta / \xydist{\min\{i,j\}}{1}{\max\{i,j\}}$, where $\delta>0$ is
such that moving every point by at most~$\delta$ does not change the
orientation of any three non-collinear points.

Let us point out that Erickson and Kim~\cite{ek-alnfc-03} describe
a construction of pairwise face-touching 3-polytopes in 3D that may
provide the basis for a different representation in our model of a complete graph.

While we have shown that all graphs admit a 3D contact
representation, these representations may be very non-symmetric and
can have very large coordinates.  This motivates the following
question and specialized 3D drawing algorithms for certain classes of
(non-planar) graphs; see the following subsections.

\begin{open}
  Is there a polynomial $p$ such that any $n$-vertex graph has a 3D
  contact representation with convex polygons on a grid of size
  $p(n)$?
\end{open}

\subsection{Bipartite Graphs}
\label{sub:bipartite}

\begin{theorem}
  \label{thm:bipartite}
  Every bipartite graph $G=(A \cup B,E)$ admits a contact
  representation by convex polygons whose vertices are restricted to
\begin{enumerate}[(a)]
\item    a toroidal grid of size $|B| \times (2|A|-2)$ or
\item a 3D integer grid of
  size $|A| \times 2\left\lceil \frac{|B|}{4} \right\rceil \times
  (\left\lceil \frac{|A|}{2} \right\rceil^2 + \left\lceil \frac{|B|}{4} \right\rceil^2)$.
\end{enumerate}  
  Such representations can be computed in $O(|E|)$ time.
\end{theorem}

\begin{proof}
  We first prove the result for complete bipartite graphs
  $K_{|A|,|B|}$.  As in the other costructions in this paper, our
  representation for $K_{|A|,|B|}$ is such that each polygon
  representing a vertex~$v$ of the graph is the convex hull of the
  touching points with adjacent polygons.  (The touching points
  represent the edges of the graph.)  If the given bipartite graph~$G$
  is not complete, we simply remove from our representation of
  $K_{|A|,|B|}$ the touching points that correspond to the non-edges
  in~$G$.

  In our construction for $K_{|A|,|B|}$, the polygons representing the
  vertices in $A$ (called \emph{$A$-polygons}) are all horizontal
  $|B|$-gons and the polygons representing the vertices in $B$ (called
  \emph{$B$-polygons}) are all vertical $|A|$-gons; see
  Fig.~\ref{fig:bipartite-tg} for an example with $|A|=|B|=8$.

  We start with a \emph{uni-monotone} convex polygon with respect to
  the z-axis (that is, a z-monotone convex polygon with a single
  segment as one of its two chains) as our \emph{lead} for generating
  a realization. For a uni-monotone polygon, we call the single edge
  monotone chain the \emph{base edge} and the other chain the
  \emph{mountain chain}. Our lead polygon represents a $B$-polygon
  (hence, is an $|A|$-gon) and has the following two properties:
  ($\expandafter{\romannumeral 1}$)~it is coplanar with the z-axis,
  and ($\expandafter{\romannumeral 2}$)~its mountain chain lies
  between the base segment and the z-axis. The remaining $B$-polygons
  are all rotated copies of the lead polygon around the z-axis (each
  with a distinct rotation angle). The $A$-polygons are horizontal and
  each at a different height, and hence they are
  interior-disjoint. $B$-polygons are also trivially
  interior-disjoint. The $A$- and $B$-polygons are interior-disjoint
  due to property ($\expandafter{\romannumeral 2}$) of our lead
  polygon. Using evenly spaced vertices on a half circle as our lead,
  we get a representation on a toroidal grid\footnote{A \emph{toroidal
      grid} of size $m \times n$ is the Cartesian product of two cycle
    graphs $C_m$ and $C_n$.} of size $|B| \times (2|A|-2)$ (note that
  our representation uses the inner half of the grid points of a
  toroidal grid because the lead polygon needs to be uni-monotone).

\begin{figure}[tb]
      \begin{subfigure}[b]{.46\linewidth} 
        \centering
        \input{figures/bipar-tp}
        \caption{projection on the xy-plane}
      \end{subfigure}  
      \hfill
      \begin{subfigure}[b]{.46\linewidth}
        \centering
        \input{figures/bipar-t}
        \caption{side view}
      \end{subfigure}
      \caption{A contact representation of $K_{8,8}$ using a toroidal grid.}
      \label{fig:bipartite-tg}
\end{figure}
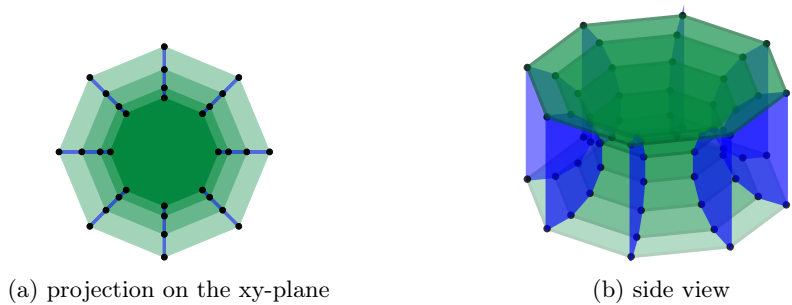

For representations on integer grids, we distort the representation above to some degree.
A \emph{core} $A$-polygon (innermost polygon in Fig.~\ref{fig:bi-ip}) is a
convex $|B|$-gon on the xy-plane using a grid of size
\[2\left\lceil \frac{|B|}{4} \right\rceil \times \left(\left\lfloor \frac{\left\lceil\frac{|B|}{2}\right\rceil}{2} \right\rfloor \left(\left\lceil \frac{|B|}{4} \right\rceil -1\right) +2 \right).\]
Recall that the grid size is the product of the number of grid lines in each dimension.
Here, we use both a core $A$-polygon, and a lead $B$-polygon (with additional properties), to generate a realization.

Our lead $B$-polygon has the following properties.
It lies on the xz-plane and is to the right of the z-axis. It has an axis of symmetry parallel to the x-axis (this helps getting a more compact representation). The z-coordinates of consecutive vertices on its mountain chain are all one unit apart. The distance between the x-coordinates of consecutive pairs of vertices along the boundary and in the direction towards the base segment increments by $1, 2, 3, \ldots, \lceil |A|/2\rceil-1$. These properties guarantee that the lead polygon is uni-monotone, convex, and has integer coordinates.
The lead $B$-polygon is incident to the leftmost vertex of the core $A$-polygon at its vertex with the minimum x-coordinate (i.e., closest to the z-axis).

The remaining $B$-polygons are again congruent, however, this time their planes are all parallel to the lead B-polygon (rather than being placed around the z-axis). This helps us maintain integer coordinates for all vertices. Figure~\ref{fig:bi-ip} shows the projection of such representations on the xy-plane. Note that the core $A$-polygon
can be split into two y-monotone chains of about the same size. The $B$-polygons that are incident to the vertices on the same y-monotone chain of the core are identical (only translated by a vector in the xy-plane). They are mirrored if they are incident to vertices on different monotone chains of the core. See Fig.~\ref{fig:bi-i} for a 3D view of an example. Our construction requires a grid of size
\[|A| \times 2\left\lceil \frac{|B|}{4} \right\rceil \times \left( \left\lceil \frac{|A|}{2} \right\rceil \left(\left\lceil \frac{|A|}{2} \right\rceil -1\right) + \left\lfloor \dfrac{\left\lceil\frac{|B|}{2}\right\rceil}{2} \right\rfloor \left(\left\lceil \frac{|B|}{4} \right\rceil -1\right) +2\right).\]
\end{proof}

\begin{figure}[tb]
        \centering
        \input{figures/bipar-ip2}
        \caption{Top view of a contact representation of $K_{8,16}$ on
          the integer grid.}
        \label{fig:bi-ip}
\end{figure}
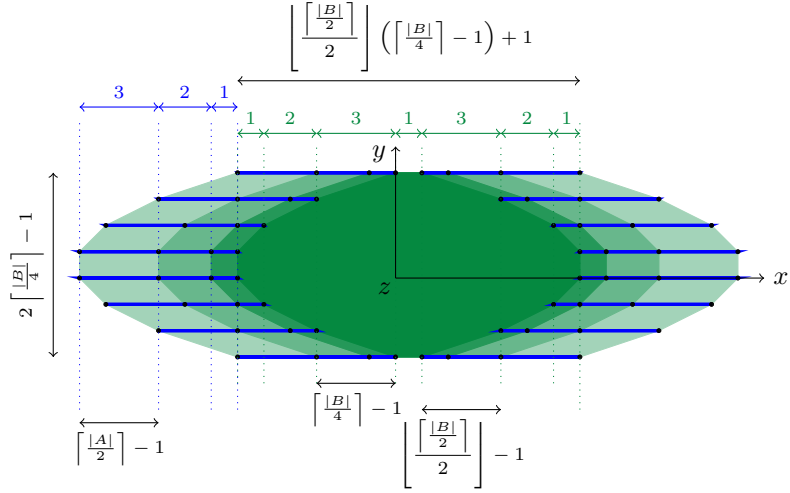

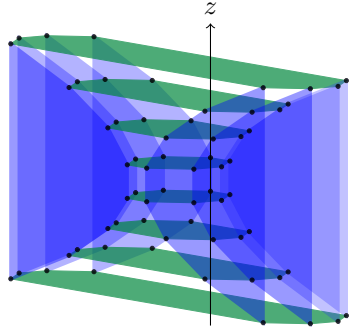
\begin{figure}[tb]
        \centering
        \input{figures/bipar-i}
        \caption{Side view of a contact representation of $K_{8,8}$ on
          the integer grid.}
        \label{fig:bi-i}
\end{figure}

\begin{proposition}
  The graph $K_{3,3}$ admits a contact representation in 3D using unit
  equilateral triangles.
\end{proposition}

\begin{proof}
  Our contact representation consists of three horizontal and three
  vertical unit equilateral triangles; see Fig.~\ref{fig:K33}(a).
  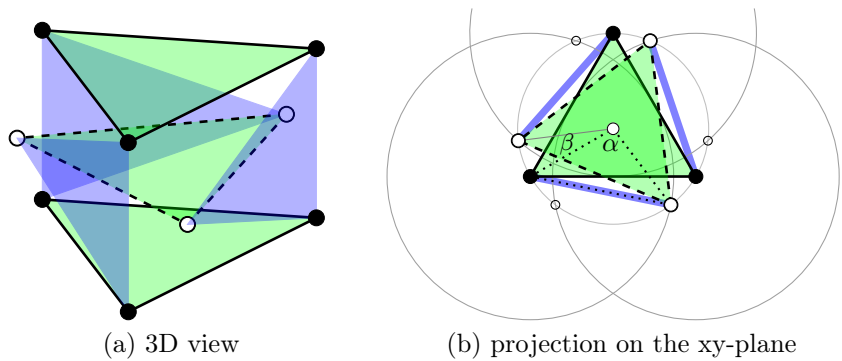
\begin{figure}[tb]
    \centering
    \begin{tabular}{c@{\qquad}c}
       \input{figures/K33} &
      \input{figures/K33C} \label{fig:k}\\
      (a) 3D view &
      (b) projection on the xy-plane\\
    \end{tabular}
    \caption{A contact representation of $K_{3,3}$ with unit equilateral
    triangles.}
    \label{fig:K33}
  \end{figure}
  The three horizontal triangles have z-coordinates~0, $1/2$, 1, and
  are centered at the z-axis.  The topmost triangle is right above the
  bottommost one, whereas the middle triangle is rotated by an
  angle~$\beta$.  In the projection on the xy-plane, all their
  vertices lie on a circle of of radius $\tan(30^\circ)$; see the
  small gray circle in Fig.~\ref{fig:K33}(b).  The figure also shows
  three big gray circles of radius $\sin(60^\circ)$ (which is the
  height of a unit equilateral triangle) centered on the vertices of
  the top- and bottommost triangles.  Each big circle intersects the
  small circle in two distinct points; in Fig.~\ref{fig:K33}(b), the
  left one is marked with a small circle, the right one with a bigger
  circle.  Connecting the right intersection points (bigger circles)
  yields the vertices of the middle horizontal triangle.
  The side lengths of the
  black dotted triangle are $\tan(30^\circ)$, $\tan(30^\circ)$, and
  $\sin(60^\circ)$.  By the law of cosines,
  $\alpha=120^\circ-\beta=\arccos(-1/8)$.  Hence, $\beta \approx 22.82^\circ$.
\end{proof}

\subsection{1-Planar Cubic Graphs}
\label{sub:1-planar}

A simple consequence of the circle-packing theorem~\cite{Koebe36} is
that every planar graph (of minimum degree 3) is the contact graph of 
convex polygons in the plane.
In this section, we consider a generalization of planar graphs called
\emph{1-planar graphs} that have a drawing in 2D in which every
edge (Jordan curve) is crossed at most once.

Our approach to realizing these graphs will use the \emph{medial
graph}~\Gmed associated with a plane graph~$G$ (or, to be more
general, with any graph that has an edge ordering).  The vertices
of~\Gmed are the edges of~$G$, and two
vertices of~\Gmed are adjacent if the corresponding edges of~$G$ are
incident to the same vertex of~$G$ and consecutive in the circular
ordering around that vertex.  The medial graph is always 4-regular.
If $G$ has no degree-1 vertices, \Gmed has no loops.  If~$G$ has
minimum degree~3, \Gmed is simple.
Also note that \Gmed is connected if and only if~$G$ is connected.

\begin{theorem}
  \label{thm:1-planar}
  Every 1-plane cubic graph with $n$ vertices can be realized as a 
  contact graph of triangles with vertices on a grid of size
  $(3n/2-1) \times (3n/2-1) \times 3$.
  Given a 1-planar embedding of the graph, it takes
  linear time to construct such a realization.
\end{theorem}

\begin{proof}
  Let $G$ be the given 1-plane graph.
  Let~$\Gmed'$ be the medial graph of~$G$ with
  the slight modification that, for each
  pair~$\{e,f\}$ of crossing edges, $\Gmed'$ has only 
  one vertex~$v_{ef}$, which is incident to all (up to eight) edges that
  immediately precede or succeed~$e$ and~$f$ in the circular order
  around their endpoints; see Fig.~\ref{fig:1-planar}\subref{fig:medial}.
  \begin{figure}[tb]
    \begin{subfigure}[b]{.38\textwidth}
      \centering
      \includegraphics[page=1]{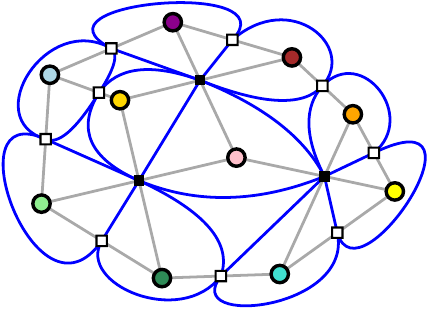}
    \end{subfigure}
    \hfill
    \begin{subfigure}[b]{.46\textwidth}
      \centering
      \includegraphics[page=2]{1-planar}
    \end{subfigure}

    \begin{subfigure}[t]{.38\textwidth}
      \caption{a 1-plane cubic graph $G$ and its (modified) medial graph
        $\Gmed'$}
      \label{fig:medial}
    \end{subfigure}
    \hfill
    \begin{subfigure}[t]{.55\textwidth}
      \caption{representation of~$G$ with triangles; the numbers
        indicate the z-coordinates of the triangle corners (unlabeled
        vertices lie in the xy-plane)}
      \label{fig:triangles}
    \end{subfigure}
    
    \caption{1-plane cubic graphs admit compact triangle
      contact representations.}
    \label{fig:1-planar}
  \end{figure}
  The order of the edges around~$v_{ef}$ is the obvious one.
  Using Schnyder's linear-time algorithm~\cite{s-epgg-SODA90} for
  drawing 3-connected graphs\footnote{%
    If $\Gmed'$ is not 3-connected, we add dummy edges to fully
    triangulate it and then remove these edges to obtain a drawing of
    $\Gmed'$.} straight-line,
  we draw~$\Gmed'$ on a planar grid of size $(3n/2-1) \times (3n/2-1)$.
  Note that this is nearly a contact representation of~$G$ except that, 
  in each crossing point, \emph{all} triangles of the respective four
  vertices touch.  Figure~\ref{fig:1-planar}\subref{fig:triangles} is a sketch of the
  resulting drawing (without using Schnyder's algorithm)
  for the graph in Fig.~\ref{fig:1-planar}\subref{fig:medial}.  

  We
  add, for each crossing $\{e,f\}$, a copy~$v'_{ef}$ of the crossing
  point~$v_{ef}$ one unit above.  Then we select an arbitrary one of
  the two
  edges, say~$e=uv$.  Finally we make the two triangles corresponding 
  to~$u$ and~$v$ incident to~$v'_{ef}$ without modifying the
  coordinates of their other vertices.  The labels in
  Fig.~\ref{fig:1-planar}\subref{fig:triangles} are the resulting z-coordinates for our
  example; all unlabeled triangle vertices lie in the xy-plane.

  If a crossing is on the outer face of~$G$, it can happen that a
  vertex of~$G$ incident to the crossing becomes the outer face
  of~$\Gmed'$; see Fig.~\ref{fig:B-configuration} where this vertex is
  called~$a$ and the crossing edges are $ac$ and $bd$.  Consider the
  triangle~$\Delta_a$ that represents~$a$ in~$\Gmed'$.  It covers the
  whole drawing of $\Gmed'$.  To avoid intersections with triangles
  that participate in other crossings, we put the vertex of~$\Delta_a$
  that represents the crossing to $z=-1$, together with the vertex of
  the triangle~$\Delta_c$ that represents~$c$.

  \begin{figure}[tb]
    \centering
    \includegraphics{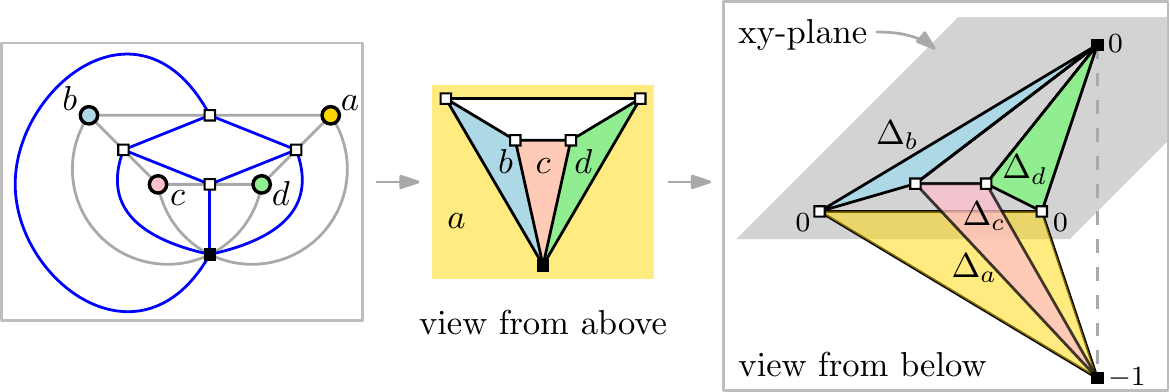}
    \caption{left: graphs $G$ (with a crossing on the outer face) and
      $\Gmed'$; center: straight-line drawing of $\Gmed'$; right:
      resulting 3D representation of~$G$ (numbers are z-coordinates).}
    \label{fig:B-configuration}      
  \end{figure}
  
  Our 3D drawing projects vertically back to the planar drawing, so
  all triangles are interior disjoint (with the possible exception of
  a triangle that represents the outer face of~$\Gmed'$).
  Triangles that share an edge
  in the projection are incident to the same crossing~-- but this
  means that at least one of the endpoints of the shared edge has a
  different z-coordinate.  Hence, all triangle contacts are
  vertex--vertex contacts.  Note that some triangles may touch each
  other at $z=1/2$ (as the two central triangles in
  Fig.~\ref{fig:1-planar}\subref{fig:triangles}), but our contact model tolerates this.
\end{proof}

\subsection{Cubic Graphs}
\label{sub:cubic}

We first solve a restricted case and then show how this helps us to
solve the general case of cubic graphs.

\begin{lemma}
  \label{lem:two-connected-cubic}
  Every 2-edge-connected cubic graph with $n$ vertices can be realized
  as a contact graph of triangles with vertices on a grid of size
  $3 \times n/2 \times n/2$.  
  It takes $O(n \log^2 n)$ time to construct such a realization.
\end{lemma}

\begin{proof}
  By Petersen's theorem~\cite{p-trg-AM1891}, any
  given 2-edge-connected cubic graph~$G$ has a perfect matching. Note that
  removing this matching leaves a 2-regular graph, i.e., a set
  of vertex-disjoint cycles $C_1,\dots,C_k$; see
  Fig.~\ref{fig:two-edge-connected}(a).  Such a partition can be
  computed in $O(n \log^2 n)$ time~\cite{ds-pmbcg-SOFSEM10}.
  Let $n=|V(G)|$ and $n_1=|V(C_1)|, \dots, n_k=|V(C_k)|$.
  Note that $n=n_1+\dots+n_k$.  We now construct a planar
  graph $H=(V,E)$ with $n+1$ vertices that will be the ``floorplan''
  for our drawing of~$G$.  The graph~$H$ consists of an $n$-wheel 
  with outer cycle 
  $v_{1,1},\dots,v_{1,n_1},\dots,v_{k,1},\dots,v_{k,n_k}$, $n$ spokes and a
  hub~$h$, with additional \emph{chords} $v_{1,1}v_{1,n_1}, v_{2,1}v_{2,n_2},
  \dots, v_{k,1}v_{k,n_k}$.  We call the edges $v_{1,n_1}v_{2,1}, \dots,
  v_{k,n_k}v_{1,1}$ \emph{dummy edges} (thin gray in
  Fig.~\ref{fig:two-edge-connected}(b) and~(c)) and the other edges on
  the outer face of the wheel \emph{cycle edges}.
  
  \begin{figure}[tb]
    \centering
    \includegraphics{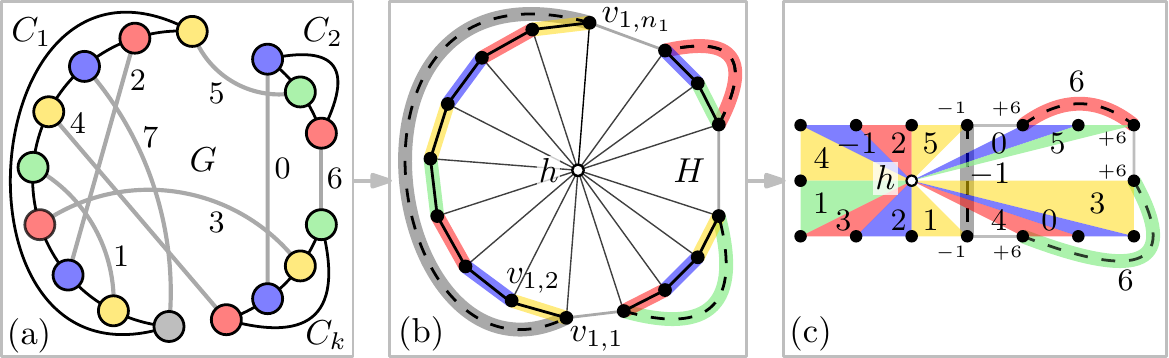}
    \caption{Representing a 2-edge-connected cubic graph~$G$ by
      touching triangles in 3D: (a)~partition of the edge set into
      disjoint cycles and a perfect matching (the numbers denote a
      permutation of the matching edges); (b)~the graph~$H$;
      (c)~3D contact representation of~$G$; the
      numbers inside the triangles indicate the z-coordinates of the
      triangle apexes (above~$h$), the small numbers denote 
      the non-zero z-coordinates of the vertices.
    }
    \label{fig:two-edge-connected}
  \end{figure}

  The chords and cycle edges form triangles with apex~$h$.  More
  precisely, for every $i \in \{1,\dots,k\}$, the chord-based triangle
  $\Delta v_{i,1}v_{i,n_i}h$ and the $n_i-1$ cycle-based triangles
  $\Delta v_{i,1}v_{i,2}h, \dots, \Delta v_{i,n_i-1}v_{i,n_i}h$
  together represent the $n_i$ vertices in the cycle~$C_i$ of~$G$.
  For each~$C_i$, we still have the freedom to choose which vertex
  of~$G$ will be mapped to the chord-based triangle of~$H$.
  This will depend on the perfect matching in~$G$.  
  The cycle edges will be drawn in the 
  xy-plane (except for those incident to a chord edge); their apexes
  will be placed at various grid points above~$h$ such that matching
  triangles touch each other.  The chord-based triangles will be drawn
  horizontally, but not in the xy-plane.

  In order to determine the height of the triangle apexes, we go
  through the edges of the perfect matching in an arbitrary order; see
  the numbers in Fig.~\ref{fig:two-edge-connected}(a).  Whenever an
  endpoint~$v$ of the current edge~$e$ is the \emph{last} vertex of a
  cycle, we represent~$v$ by a triangle with chord base.  We place the
  apexes of the two triangles that represent~$e$ at the lowest free
  grid point above~$h$; see the numbers in
  Fig.~\ref{fig:two-edge-connected}(c).  Our placement ensures that,
  in every cycle (except possibly one, to be determined later), the
  chord-based triangle is the topmost horizontal triangle;
  all cycle-based triangles are below it.  This guarantees that
  the interiors of no two triangles intersect (and the triangles of
  adjacent vertices touch).

  Now we remove the chords from~$H$.  The resulting graph is a wheel;
  we can simply draw the outer cycle using grid points on the
  boundary of a $(3 \times n/2$)-rectangle 
  and the hub on any grid point in the interior.  
  (For the smallest cubic graph, $K_4$, we would actually need a
  $(3\times3)$-rectangle, counting grid lines, in order to have a
  grid point in the interior, but it's not hard to see that $K_4$ can
  be realized on a grid of size $3 \times 2 \times 2$.) 
  If one of the $k$ cycles encloses~$h$ in the drawing (as~$C_1$ in
  Fig.~\ref{fig:two-edge-connected}(c)), we move its chord-based
  triangle from $z=z^\star>0$ to the plane $z=-1$, that is, below all
  other triangles.  Let~$i^\star$ be the index of this cycle (if it
  exists).  Note that this also moves the apex of the triangle that is
  matched to the chord-based triangle from $z=z^\star$ to~$z=-1$.  In
  order to keep the drawing compact, we move each apex with
  z-coordinate $z'>z^\star$ to $z'-1$.  Then the height of our drawing
  equals exactly the number of edges in the perfect matching, that is,
  $n/2$.

  The correctness of our representation follows from the fact that, in
  the orthogonal projection onto the xy-plane, the only pairs of
  triangles that overlap are the pairs formed by a chord-based
  triangle with each of the triangles in its cycle and, if it exists,
  the chord-based triangle of~$C_{i^\star}$ with all triangles of the
  other cycles.  Also note
  that two triangles $\Delta v_{i,j-1}v_{i,j}h$ and
  $\Delta v_{i,j}v_{i,j+1}h$ (the second indices are modulo~$n_i$)
  that represent consecutive vertices in~$C_i$ (for some
  $i \in \{1,\dots,k\}$ and $j \in \{1,\dots,n_i\}$) touch only in a
  single point, namely in the image of~$v_{i,j}$.  This is due to the
  fact that vertices of~$G$ that are adjacent on~$C_i$ are not
  adjacent in the matching, and for each matched pair its two triangle
  apexes receive the same, unique z-coordinate.

  We do not use all edges of~$H$ for our 3D contact
  representation of~$G$.  The spokes of the wheel are the projections
  of the triangle edges incident to~$h$.  The $k$ dummy edges don't
  appear in the representation (but play a role in the proof of
  Theorem~\ref{thm:cubic} ahead).
\end{proof}

In order to generalize Lemma~\ref{lem:two-connected-cubic} to any
cubic graph~$G$, we use the \emph{bridge-block tree} of~$G$.  This
tree has a vertex for each 2-edge-connected component and an edge for
each bridge of~$G$.  The bridge-block tree of a graph can be computed
in time linear in the size of the
graph~\cite{wt-mbcbc-Algorithmica92}.
The general idea of the construction is the following. First, remove all 
bridges from $G$ and, using some local replacements, transform each connected component 
of the obtained graph into a 2-edge-connected cubic graph.  Then, use
Lemma~\ref{lem:two-connected-cubic}
to construct a representation of each of these graphs.
Finally, modify the obtained representations to undo the local
replacements and use the
bridge-block tree structure to connect the constructed subgraphs, restoring the bridges of~$G$.

\begin{theorem}
  \label{thm:cubic}
  Every cubic graph with $n$ vertices can be realized as a contact
  graph of triangles with vertices on a grid of size
  $3n/2 \times 3n/2 \times n/2$.  It takes $O(n \log^2 n)$ time to construct
  such a realization.
\end{theorem}

\begin{proof}
  We can assume that the given graph~$G$ is connected, otherwise we
  draw each connected component separately and place the drawings
  side-by-side.  Then the bridge-block tree of~$G$ yields a partition
  of~$G$ into 2-edge-connected components $G_1,\dots,G_k$, which are
  connected to each other by bridges.

  We go through~$G_1,\dots,G_k$ and construct
  floorplan graphs~$H_1,\dots,H_k$ as follows.  If~$G_i$ is a
  single vertex, let~$H_i$ be a triangle.  For an example, see~$H_6$
  in Fig.~\ref{fig:general-cubic-graphs}.  
  If a component $G_i$ with $n_i>1$ doesn't contain any matching edge,
  that is, if all its vertices are endpoints of bridges, then
  let~$H_i$ be (an internally triangulated) $n_i$-cycle.  The vertices
  in~$G_i$ will be represented by triangles whose bases are the edges
  of the cycle and whose apexes lie outside the cycle.  Each apex~$v_b$
  corresponds to a bridge~$b$ and will later be connected to a triangle
  representing the other endpoint of the bridge.

  Otherwise, we remove each
  vertex in~$G_i$ that is incident to a bridge and connect its two
  neighbors, so that we can apply Petersen's
  theorem~\cite{p-trg-AM1891} to~$G_i$.  We
  call the new edge the \emph{foot} of the bridge.  This yields a
  collection of cycles and a perfect matching in~$G_i$.  As in the
  proof of Lemma~\ref{lem:two-connected-cubic}, $H_i$ is a wheel
  with $|V(G_i)|+1$ vertices, and we compute, for each component~$G_i$,
  the heights of all triangle apexes.  This also determines which
  vertices of~$G_i$ are represented by chord-based triangles.
  If applying Petersen's theorem to~$G_i$ gives rise to a single
  cycle, we consider the chord (which will be drawn at
  $z=-1$) to simultaneously be a dummy edge (which will be ``drawn''
  at $z=0$), so that every graph~$H_i$ has a dummy edge.
  (For examples, see~$H_3$ or~$H_5$ in
  Fig.~\ref{fig:general-cubic-graphs}.)

  Let~$H$ be the disjoint union of~$H_1,\dots,H_k$.  Now we
  reintroduce the bridges.  For every bridge~$b$, we add a new
  vertex~$v_b$ to~$H$.
  Each foot
  of~$b$ is either a cycle edge
  or a matching edge in some~$G_i$, which we treat differently;
  see Fig.~\ref{fig:general-cubic-graphs}.
  
  \begin{figure}[tb]
    \centering
    \includegraphics{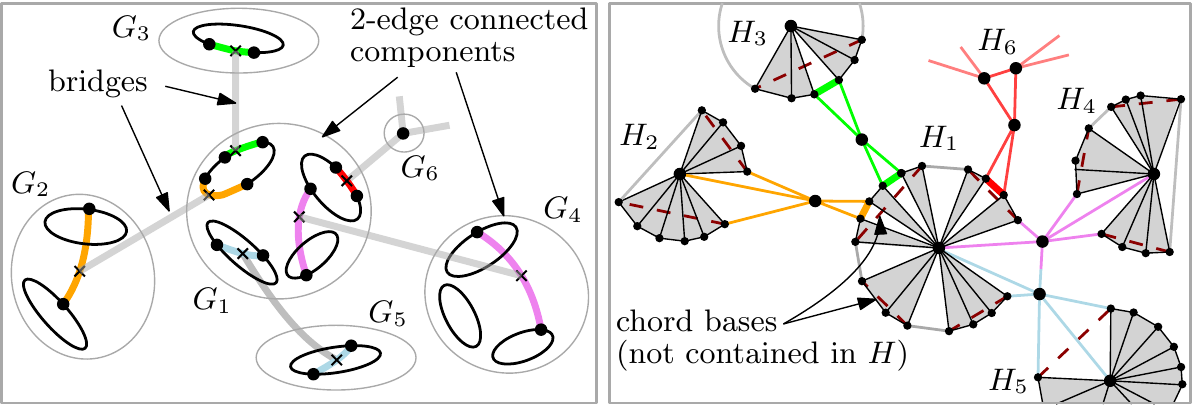}
    \caption{Constructing the floorplan~$H$ of a general cubic graph}
    \label{fig:general-cubic-graphs}
  \end{figure}

  If the foot~$uw$ of~$b$ is a cycle edge, consider the two adjacent
  triangles in~$H_i$
  that share the vertex representing the foot $uw$.
  These triangles
  share the hub~$h_i$ of~$H_i$
and a vertex~$v_{uw}$ on the outer face
  of~$H_i$.  We take the two triangles apart by duplicating~$v_{uw}$.
  We connect each copy of~$v_{uw}$ to the other copy, to~$v_b$,
  to~$h_i$, and to a different neighbor along the cycle.  The new
  edges between the two copies and between them and~$v_b$ form a
  triangle that represents one of the two endpoints of the bridge~$b$;
  see Fig.~\ref{fig:bridge-types} (right). 

  If the foot~$uw$ of~$b$ is a matching edge, we pick a dummy
  edge~$xy$ on the outer face of~$H_i$.  Recall that dummy edges are
  the edges that connect the cycles in~$H_i$ (thin gray in
  Figs.~\ref{fig:two-edge-connected}(b) and~(c)).  Due to our
  construction, $H_i$ contains at least one dummy edge.  We remove the
  dummy edge~$xy$ and connect~$x$, $h_i$, and $y$ to~$v_b$ in this order.  Note that
  several bridge feet can be placed into the space reserved by a single
  dummy edge (see the bridges that connect~$H_4$ and~$H_5$
  to~$H_1$ in Fig.~\ref{fig:general-cubic-graphs} (right)).
  
  Then we draw~$H$ in the xy-plane, using Schnyder's linear-time
  algorithm~\cite{s-epgg-SODA90}.  (In order to make~$H$ 3-connected,
  we add edges in the outer face of~$H$ that connect the components
  that are leaves of the bridge-block tree.)  Finally, as in the proof
  of Lemma~\ref{lem:two-connected-cubic}, we insert the chord edges
  (at the correct heights) and extend all cycle and chord edges into
  triangles by placing their apexes at the locations above or
  below~$h_i$ that we've computed before.
  Whenever we place two apexes that correspond to a matching edge that
  is the foot of a bridge~$b$, we use two consecutive grid points, one
  for each apex.  (If one of the apexes belongs to the chord-based
  triangle at $z=-1$, we place the other apex at $z=0$.) 
  Together with~$v_b$ (which remains on the xy-plane), the two apexes form
  a vertical triangle; see Fig.~\ref{fig:bridge-types} (left).
  The projection of the triangle to the xy-plane is an edge of~$H$;
  the vertical (closed) slab above that edge is used exclusively by
  the new triangle. 

  To bound the grid size of the drawing, we show that $|V(H)| \le
  |E(G)|$ ($=3n/2$), by establishing an injective map from $V(H)$
  to~$E(G)$: we map every cycle vertex in~$H$ to the ccw next cycle
  edge in~$H$, which corresponds to a specific cycle edge in~$G$.
  Further, we map every bridge vertex~$v_b$ to the corresponding
  bridge~$b$ in~$G$.  It remains to map the hubs.  If a
  component~$H_i$ of~$H$ does not contain any matching edge (that is,
  all vertices in~$H_i$ are incident to bridges), $H_i$ does not
  contain a hub.  Otherwise, there is at least one
  matching edge in~$H_i$ and we map the hub~$h_i$ to that edge.

  Now it is clear that the straight-line drawing of~$H$ computed by
  Schnyder's algorithm has size at most $(3n/2-1) \times (3n/2-1)$.
  In order to bound the height of the drawing, consider any
  component~$H_i$ of~$H$.  Clearly, $H_i$ contains at most $n/2$
  matching edges.  Each of these uses a grid point on the vertical
  line through~$h_i$.  Any matching edge can, however, be the foot of
  a bridge.  For each bridge triangle that we insert between the
  apexes of two matching triangles, the height of the representation
  of~$H_i$ increases by one unit.  On the other hand, the bridges form
  a matching that is independent from the matching edges.  Thus, the
  height of~$H_i$ is at most $n/2$.
\end{proof}

\begin{figure}[tb]
  \centering
  \includegraphics{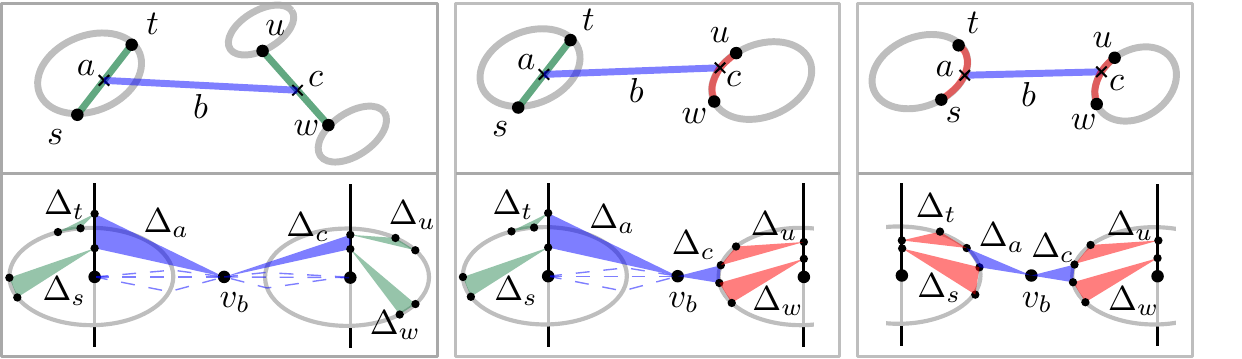}
  \caption{Representation of a bridge $b=ac$ depending on the types of its feet}
  \label{fig:bridge-types}
\end{figure}

\begin{corollary}
  \label{cor:cubic}
  Every graph with $n$ vertices and maximum degree~3 can be realized
  as a contact graph of triangles, line segments, and points whose
  vertices lie on a grid of size $3\lceil n/2 \rceil \times 3 \lceil
  n/2 \rceil \times \lceil n/2 \rceil$.  It takes $O(n \log^2 n)$ time
  to construct such a realization.
\end{corollary}

\begin{proof}
  If $n$ is odd, add a dummy vertex to the given graph.  Then add
  dummy edges until the graph is cubic.  Apply
  Theorem~\ref{thm:cubic}.  From the resulting representation, remove
  the triangle that corresponds to the dummy vertex, if any.
  Disconnect the pairs of triangles that correspond to dummy edges.
\end{proof}

\subsection{Squares of Cycles}
\label{sub:squares}

Recall that, for an undirected graph~$G$ and an integer $k \ge 2$, the
$k$-th \emph{power} $G^k$ of~$G$ is the graph with the same vertex set
where two vertices are adjacent when their distance in~$G$ is at
most~$k$.
Note that $C_4^2=K_4$ is 3-regular (and can be represented by four unit
equilateral triangles that pairwise touch and form an octahedron with
four empty faces).  For $n \ge 5$, $C_n^2$ is 4-regular.
Recall that Corollary~\ref{cor:general} yields contact representations
using convex polygons for any graph, but in these representations the
ratio between the length of the longest edge and the length of the
shortest edge can be huge.  For squares of cycles, we can do better.

\begin{theorem}
  \label{thm:squares-of-cycle}
  For $n \ge 5$, $C_n^2$ admits a contact representation using convex
  quadrilaterals in 3D such that the ratio between the length of the
  longest edge and the length of the shortest edge over all
  quadrilateral edges in the representation is constant.  For even $n$, the
  quadrilaterals can be unit squares.
\end{theorem}

\begin{proof}
  For even $n \ge 6$, $C_n^2$ is planar, so it is easy to find a
  contact representation with convex quadrilaterals in the plane.  In 3D,
  however, we can restrict the quadrilaterals to unit squares; see
  Fig.~\ref{fig:c2}\subref{fig:c2_67}.  Note that the vertices on the
  middle plane form a regular $n$-gon and the vertices on the top and
  bottom planes form regular $(n/2)$-gons, all centered at the z-axis.
  Additionally, each vertex of the top or bottom plane lies on the
  bisector of the (empty) triangular face incident to it.

  Finding a representation of $C^2_5=K_5$ we leave as an exercise to
  the reader.  Now we obtain a representation of~$C^2_{n+1}$ from that
  of~$C^2_n$ for even $n \ge 6$ (see Fig.~\ref{fig:c2} for an
  illustration of our construction for $n=6$).  Every vertex~$v_i$ in
  $C^2_{n+1}$ is represented by a quadrilateral~$Q_i$.  For two
  touching quadrilaterals~$Q_i$ and~$Q_j$, let~$P_{i,j}$ be their
  contact point (which represents the edge $v_iv_j$ in the
  graph~$C^2_{n+1}$).  Note that all edges of $C^2_n$ except for two,
  namely $v_1v_{n-1}$ and $v_2v_n$, are also edges in $C^2_{n+1}$.
  Moreover, the endpoints of these two non-edges of $C^2_{n+1}$ are
  all incident to $v_{n+1}$.  Therefore, we can easily extend a
  contact representation of $C^2_n$ into one of~$C^2_{n+1}$ as
  follows.  Starting with a contact representation of $C^2_n$, we
  duplicate the vertices $T=P_{1,n-1}$ and $B=P_{2,n}$ (red in
  Fig.~\ref{fig:c2}\subref{fig:c2_67}) to separate the pairs
  $(Q_1,Q_{n-1})$ and $(Q_2,Q_{n})$ of touching quadrilaterals such
  that the four new vertices form the new quadrilateral~$Q_{n+1}$.

  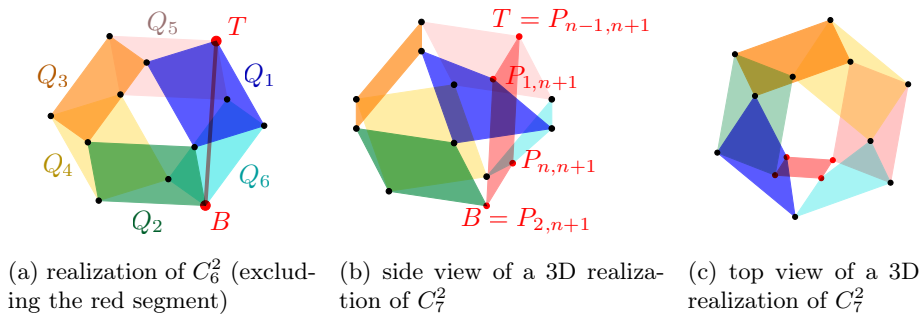
\begin{figure}[tb]
    \captionsetup[subfigure]{position=b}
    \centering
    \subcaptionbox{realization of $C^2_6$ (excluding the red
      segment)\label{fig:c2_67}}{\hspace*{-1ex}\input{figures/c2_6}~~~}\hfill%
    \subcaptionbox{side view of a 3D realization of
      $C^2_7$\label{fig:c2_7-n}}{\hspace*{-1.5ex}\input{figures/c2_7-s}}\hfill%
    \subcaptionbox{top view of a 3D realization of
      $C^2_7$\label{fig:c2_7-u}}{\raisebox{0.3cm}{\input{figures/c2_7-u}}} 
    \caption{Representing squares of cycles: we build a contact representation
      of~$C^2_7$ from that of~$C^2_6$.  To this end, we split the two
      big (red) vertices to expand the segment connecting them into a
      rectangle.}
    \label{fig:c2}
  \end{figure}

  In the following, we detail how we place the new vertices.  We place
  $P_{n-1,n+1}$ and $P_{2,n+1}$ at the position of $T$ and $B$
  respectively (thus $Q_{n-1}$ and $Q_2$ do not change).  Consider the
  four points~$T$, $B$, and the centers of~$Q_1$ and~$Q_n$.  Due to
  the symmetry of our representation for $C^2_n$, these four points
  span a plane~$H$.  Starting in~$T$, we move $P_{1,n+1}$ along the
  line $H\cap Q_1$.  We stop just before we reach the center of~$Q_1$.
  Symmetrically, we define the position of $P_{n,n+1}$ along the line
  $H\cap Q_n$ near the center of~$Q_n$.  Stopping the movement before
  reaching the centers of~$Q_1$ and~$Q_n$ makes sure that these faces
  remain strictly convex.  In the resulting representation of
  $C^2_{n+1}$, every
  quadrilateral
  edge has length at
  most~$|\overline{TB}|<2$ and at least (nearly) $\sqrt{2}/2$, which
  implies that the ratio of the longest edge length and the shortest
  edge length is constant.
\end{proof}

\section{Hypergraphs}
\label{sec:hypergraphs}

We start with a negative result.
Hypergraphs that give rise to simplicial 2-complexes that are not
embeddable in 3-space also do not have a realization using touching
polygons.  Carmesin's example of the cone over the complete graph
$K_5$ is such a 2-complex\footnote{Carmesin~\cite{Carmesin19} credits
  John Pardon with the observation that the \emph{link graph} at a
  vertex~$v$, which contains a node for every edge at~$v$ and an arc
  connecting two such nodes if they share a face at~$v$, must be
  planar for the 2-complex to be embeddable.\label{ftn:carmesin}},
which arises from the
3-uniform hypergraph on six vertices whose edges are
$\{\{i,j,6\} \colon \{i,j\} \in [5]^2 \}$.  Recall that $d$-uniform
means that all hyperedges have cardinality~$d$.
Any 3-uniform hypergraph that contains these edges also cannot be drawn.
For example, $\K^d_n$, the complete $d$-uniform hypergraph on $n \ge
6$ vertices for $d = 3$ does not have a non-crossing drawing  in 3D.
Note that in complete hypergraphs many pairs of hyperedges share two
vertices.  This motivates us to consider 3-uniform
\emph{linear} hypergraphs, i.e., hypergraphs where pairs of edges
intersect in at most one vertex.
Very symmetric examples of such hypergraphs are \emph{Steiner
systems} (definition below).

\subsection{Representing Steiner Systems by Touching Polygons}
\label{app:steiner}

A Steiner system $S(t,k,n)$ is an $n$-element
set~$S$ together with a set of $k$-element subsets of~$S$ (called
\emph{blocks}) such that each $t$-element subset of~$S$ is contained
in exactly one block.
In particular, Steiner triple systems $S(2,3,n)$ are examples of
3-uniform hypergraphs on $n$ vertices; see
Table~\ref{tab:SteinerSys}~\cite{w-sts-19}.
They exist for any $n \in \{6k+1, 6k+3 \colon k \in \mathbb{N}\}$.
The corresponding 3-uniform hypergraph has $n(n-1)/6$
hyperedges and is $((n-1)/2)$-regular.

First we show that the two smallest triple systems, i.e., $S(2,3,7)$
(also called the \emph{Fano plane}) and $S(2,3,9)$, admit non-crossing
drawings in 3D.  The existence of
such representations also follows from Ossona de Mendez'
work~\cite{o-rp-JGAA02} (see introduction) since
both hypergraphs have incidence orders of dimension~4 (which can be
checked by using an integer linear program).  
While our drawings have good vertex resolution (ratio between the
smallest and the longest vertex--vertex distance) and show symmetries,
Ossona de Mendez uses coordinates $1,d+1,(d+1)^2,\dots,(d+1)^{n-1}$ in
each of $d=4$ dimensions and then projects this 4D contact
representation centrally onto a 3D hyperspace.  The resulting 3D
contact representation does not lie on a grid.

\begin{proposition}
  \label{prop:fano}
  The Fano plane $S(2,3,7)$ and the Steiner triple system $S(2,3,9)$
  admit non-crossing drawings using triangles in~3D.
\end{proposition}

\begin{proof} 
  We first describe our construction for the Fano plane, which has
  seven vertices and seven hyperedges; see Table~\ref{tab:SteinerSys}
  and Fig.~\ref{fig:fano}.  We start with a unit equilateral triangle
  on the xy-plane centered at the z-axis representing hyperedge $642$
  (with vertices in ccw-order).
  We make a copy of this triangle, lift it by one unit, and rotate it
  by an angle of $\alpha$  counterclockwise around the z-axis, where
  $0^\circ < \alpha < 120^\circ$ and $\alpha \neq 60^\circ$
  (Fig.~\ref{fig:fano} uses $\alpha=85^\circ$). The copied triangle
  is not a hyperedge but determines the position of vertices~$3$, $5$,
  and~$7$ (i.e., after the transformation, vertices~$6$, $4$, and $2$ are mapped to vertices $3$, $5$, and $7$, respectively).  We place vertex~$1$ at $(0,0,1/2)$.

\begin{table}[tb]
  \centering
  \caption{The two smallest Steiner triple and quadruple systems}
  \label{tab:SteinerSys}

  \small
  \begin{tabular}[t]{cp{3mm}c}
    \toprule
    \multicolumn{1}{c}{$S(2,3,7)$}\\
    \midrule
    1 2 3\\
    1 4 7\\
    1 5 6\\
    2 4 6 \\    
    2 5 7\\
    3 4 5\\
    3 6 7 \\
    \bottomrule
  \end{tabular}
  \hfill
  \begin{tabular}[t]{cc}
    \toprule
    \multicolumn{2}{c}{$S(2,3,9)$}\\
    \midrule
    1 2 3 & 1 5 9\\
    4 5 6 & 2 6 7\\
    7 8 9 & 3 4 8\\
    1 4 7 & 1 6 8\\
    2 5 8 & 2 4 9\\
    3 6 9 & 3 5 7\\
    \bottomrule
  \end{tabular}
  \hfill
  \begin{tabular}[t]{cc}
    \toprule
    \multicolumn{2}{c}{$S(3,4,8)$}\\
    \midrule
    1 2 4 8 & 3 5 6 7\\
    2 3 5 8 & 1 4 6 7\\
    3 4 6 8 & 1 2 5 7\\
    4 5 7 8 & 1 2 3 6\\
    1 5 6 8 & 2 3 4 7\\
    2 6 7 8 & 1 3 4 5\\
    1 3 7 8 & 2 4 5 6\\
    \bottomrule
  \end{tabular}
  \hfill
   \begin{tabular}[t]{ccc}
     \toprule
     \multicolumn{3}{c}{$S(3,4,10)$}\\
     \midrule
     1 2 4 5 & 1 2 3 7 & 1 3 5 8\\
     2 3 5 6 & 2 3 4 8 & 2 4 6 9\\
     3 4 6 7 & 3 4 5 9 & 3 5 7 0\\
     4 5 7 8 & 4 5 6 0 & 1 4 6 8\\
     5 6 8 9 & 1 5 6 7 & 2 5 7 9\\
     6 7 9 0 & 2 6 7 8 & 3 6 8 0\\
     1 7 8 0 & 3 7 8 9 & 1 4 7 9\\
     1 2 8 9 & 4 8 9 0 & 2 5 8 0\\
     2 3 9 0 & 1 5 9 0 & 1 3 6 9\\
     1 3 4 0 & 1 2 6 0 & 2 4 7 0\\
     \bottomrule
   \end{tabular}
\end{table}

\begin{figure}[tb]
      \begin{subfigure}[b]{.48\linewidth} 
        \centering
        \input{figures/fano-2d}
      \end{subfigure}  
      \hfill
      \begin{subfigure}[b]{.48\linewidth}
        \centering
        \input{figures/fano-t2}
      \end{subfigure}

      \caption{The Fano plane and a drawing using touching triangles
        in~3D}
      \label{fig:fano}
\end{figure}
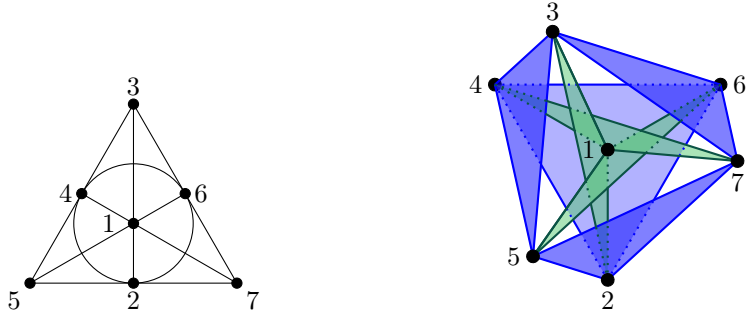

  The three (green) triangles sharing vertex~$1$ are interior-disjoint for any $\alpha$ and are non-degenerate for $\alpha \neq 60^\circ$.
  For $0^\circ < \alpha < 120^\circ$, the four (blue) triangles that are
  not incident to~$1$ are interior-disjoint and intersect no other triangles.

  Now we turn to $S(2,3,9)$; see Table~\ref{tab:SteinerSys} and 
Fig.~\ref{fig:st9-c}. We start with a unit equilateral triangle
on the xy-plane centered at the z-axis representing hyperedge $852$
(with vertices in ccw-order). We make a copy of this triangle, lift it up
by one unit, rotate it by an angle $\beta$ counterclockwise around the z-axis, and
scale it from its center by the factor $1/5$. This gives us triangle $369$.
We place vertices $1$ and $4$ at $(0,0,3/4)$ and $(0,0,1/4)$, respectively.
Figure~\ref{fig:st9-p} illustrates the triangles induced by these eight vertices.

It is easy to see that for any $\beta \le 60^\circ$, the (blue)
triangles incident to vertex~$4$ are interior-disjoint. Suppose
$\beta=60^\circ$. Then, in the projection on the xy-plane (through the
z-axis), the (green) triangles incident to vertex~$1$ map to three
segments all intersecting at the same point. Thus, in order for the
green triangles to be interior-disjoint, we need $\beta < 60^\circ$;
and in fact the smaller the scale factor is (from one), the smaller
$\beta$ needs to be.
Note that $\beta$ cannot be too small, as otherwise the blue and green triangles would intersect. More precisely, $\beta$ should be large enough so that for any two triangles $4uv$ and $1uw$, where $u \in \{2,5,8\}$ and $v,w \in \{3,6,9\}$, in the projection through the directed line $u4$, the projection of $v$ is to the right of the projection of $w$ and the projection of $w$ is to the right of the projection of $1$. In our construction, we use $\beta=45^\circ$ (and the scale factor $1/5$), which satisfies all the required conditions for having the blue and green triangles interior-disjoint.
So far, we have determined the position of eight vertices (i.e., all but $7$) such that the (eight) triangles induced by them are all pairwise non-intersecting.

Vertex $7$ (not placed yet) forms four triangles with segments $26$,
$35$, $89$, and $14$.  Note that the first three of these segments
are on the convex hull of the vertices put so far (see Fig.~\ref{fig:st9-p}\subref{fig:st9-ps}).
Let $\ell$ denote the intersection line of the planes defined by $358$ and $269$. The projection of $\ell$ and the projection of segment $89$ on the xy-plane intersect (see Fig.~\ref{fig:st9-p}\subref{fig:st9-pp}).  
Let $H$ be the plane containing $89$ and parallel to the z-axis, and
let $P$ be the intersection point of $H$ and $\ell$. If vertex $7$ is
above $P$, then the set of eight triangles not incident to $7$,
together with triangles $267$, $357$, and $789$ are all pairwise
non-intersecting (recall that segments $26$, $35$, and $67$ are on the
convex hull of vertices $\{1,\dots,6,8,9\}$). In order to make sure
that triangle $147$ does not cause any intersections, we fix the
position of $7$ by lifting $P$ slightly so that it still remains below
the plane defined by $123$. In our construction, $7$ is obtained by
lifting $P$ by $1/10$ unit. Our drawing fits in a rectangular cuboid
of size $1 \times 1 \times 1.5$. The minimum distance between any two
vertices is $1/5$ in this drawing. Figure~\ref{fig:st9-c} illustrates
a (complete) representation of $S(2,3,9)$.
\end{proof}

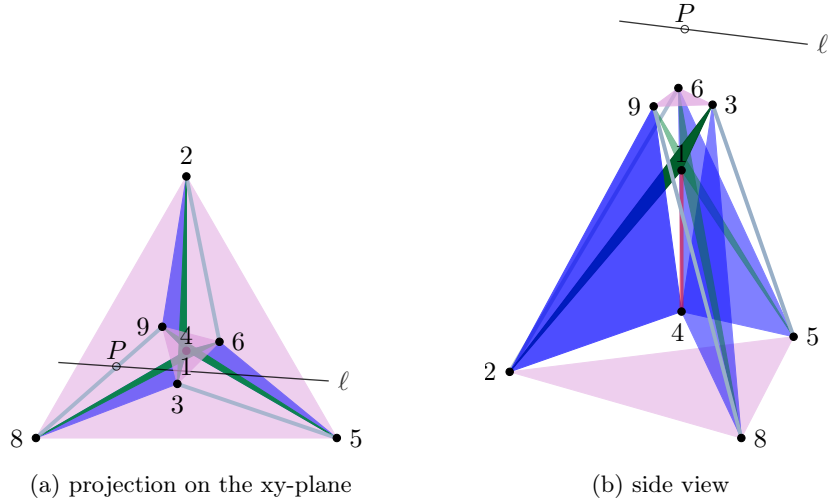
\begin{figure}[tb]
	\centering  
	\begin{subfigure}[b]{.55\linewidth}
		\centering
		\input{figures/n-st9-pp}
		\caption{projection on the xy-plane}
		\label{fig:st9-pp}  
	\end{subfigure}  
	\hfill
	\begin{subfigure}[b]{.42\linewidth}
	    \centering
            \input{figures/n-st9-ps}
	    \caption{side view}
	    \label{fig:st9-ps}
        \end{subfigure}
          	\caption{Partial drawing of the Steiner triple system $S(2,3,9)$ showing triangles not incident to $7$, together with the opposite segments to $7$ for the remaining triangles.}
	\label{fig:st9-p}
\end{figure}

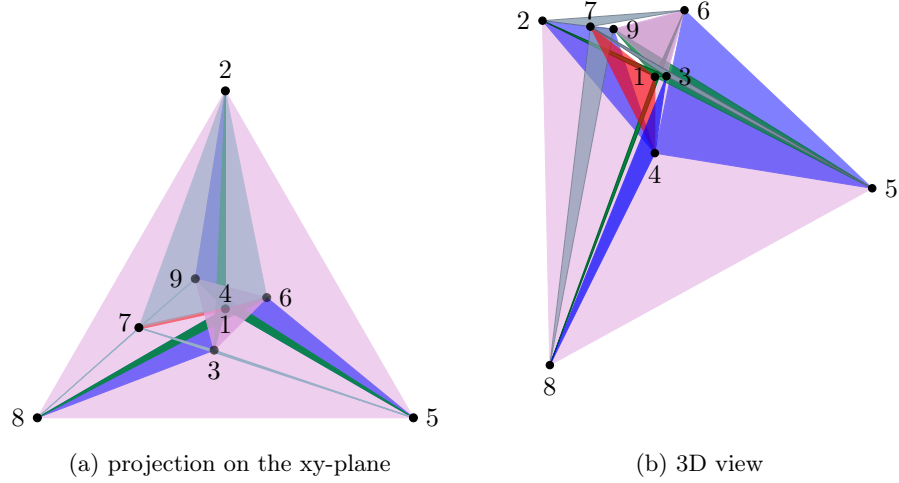
\begin{figure}[tb]
	\centering  
	\begin{subfigure}[b]{.55\linewidth}
		\centering
		\input{figures/n-st9-cp}
		\caption{projection on the xy-plane}
		\label{fig:st9-cp}  
	\end{subfigure}  
	\hfill
	\begin{subfigure}[b]{.42\linewidth}
	    \centering
            \input{figures/n-st9-cs}
	    \caption{3D view}
	    \label{fig:st9-cs}
        \end{subfigure}
   	\caption{3D contact representation of the Steiner triple
          system $S(2,3,9)$.}
	\label{fig:st9-c}
\end{figure}

Now we turn to a special class of 4-uniform hypergraphs; Steiner
quadruple systems $S(3,4,n)$~\cite{w-sqs-19}.  They exist for any
vertex number in $\{6k+2, 6k+4 \colon k \in \mathbb{N}\}$.
For $n = 8, 10, 14, \dots$, the corresponding 4-uniform hypergraph
has $m=\binom{n}{3}/4$ hyperedges and vertex degree $4m/n=(n-1)(n-2)/6$.
In the following, we study the realizability of Steiner quadruple
systems.

\begin{observation}
  \label{obs:4inPlane}
  In a non-crossing drawing of a Steiner quadruple system using
  quadrilaterals in~3D, every plane contains at most four vertices.
\end{observation}

\begin{proof}
  Suppose that there is a drawing~$R$ and a plane~$\Pi$ that
  contains at least five vertices.
  Let $ab$ be a maximum length edge of the convex hull of the points
  in the plane $\Pi$.
  No four, say $wxyz$ in that order, can be collinear, otherwise the
  quadrilateral containing $wyz$ is either $wxyz$, which is degenerate
  (a line segment), or it contains $x$ on its perimeter but $x$ is
  not a corner, a contradiction.
  Thus the set $S$ of vertices on $\Pi$ that are
  not on the edge $ab$ has size at least two.
  If there exist $u,v \in S$ such that $abu$ and $abv$
  form\footnote{In a Steiner quadruple system, every triple of
    vertices appears in a unique quadruple.} two distinct
  quadrilaterals with $ab$ then these quadrilaterals intersect
  in the plane (they are both on the same side of $ab$), a contradiction.
  If no such pair exists then $S$ contains exactly two points and they
  form one quadrilateral with $ab$, which must contain the other
  vertex in $\Pi$ (on the edge $ab$) that is not a corner, a contradiction.
\end{proof}

Observation~\ref{obs:4inPlane} is the starting point for the following result.

\begin{proposition}
  \label{thm:noS348asQuads}
  The Steiner quadruple system $S(3,4,8)$ does not admit a
  non-crossing drawing using (convex or non-convex) quadrilaterals
  in~3D.
\end{proposition}

\begin{proof}
  The Steiner quadruple system $S(3,4,8)$ has eight vertices and $14$
  hyperedges and is unique; see Table~\ref{tab:SteinerSys}.

  Assume that $S(3,4,8)$ has a contact representation by
  quadrilaterals. Without loss of generality, assume that
  quadrilateral $1248$ lies on the xy-plane.
  We show that the supporting plane of the triple $367$ is also the
  xy-plane, which, by Observation~\ref{obs:4inPlane}, is a
  contradiction.

  The line through $18$ and the line through $24$ either intersect in
  a point $v$ on the xy-plane or are parallel.
  The supporting planes of $1378$ and $2347$ both contain the line
  through $37$ and the point $v$ or, if $v$ doesn't exist, the line
  $37$ is parallel to $18$ and $24$.
  Similarly, the lines $14$ and $28$ intersect in a point $w$ on the
  xy-plane or are parallel.
  The supporting planes of $1467$ and $2678$ both contain the line
  through $67$ and the point $w$ or, if $w$ doesn't exist, the line
  $67$ is parallel to $14$ and $28$.
  Again, a similar statement holds for the intersection $u$ on the
  xy-plane of the lines $12$ and $48$.
  The supporting planes of $1236$ and $3468$ both contain the line
  $36$ and the point $u$ or, if $u$ doesn't exist, the line $36$ is
  parallel to $12$ and $48$.  These conditions imply that the
  supporting plane of $367$ is parallel to the xy-plane 
  (unless $3$, $6$, and $7$ are all collocated which is not possible as otherwise quadrilateral $3567$ is a segment).
  Since at least one of $u$, $v$, and $w$ exists and is in the
  xy-plane, $367$ lies in the xy-plane,
  contradicting Observation~\ref{obs:4inPlane}.
\end{proof} 

The main observation used in proving Proposition~\ref{thm:noS348asQuads} is that if we partition any quadruple $abcd$ into two pairs in any way, there exists a fixed pair, say $ef$, such that the union of $ef$ and each of the two partitions form a quadruple in $S(3,4,8)$. We note that the same property holds for $S(3,4,10)$. However, since this case contains more vertices, the ``fixed pairs'' obtained from different ways of ``partitioning'' of a quadruple would not have common vertices, and hence this property alone is not enough 
to show that $S(3,4,10)$ cannot be realized.  We show this for a
restricted case using the following auxiliary lemma.

\begin{lemma}
  \label{lem:quadruple}
  Let $H$ be a hypergraph whose edge set contains the subset $F =
  \{abcd, abuv, cduv, acwx, bdwx, adyz, bcyz\}$. Then $H$ does not admit a
  contact representation by quadrilaterals in which the edges in $F$
  are all convex or all non-convex. 
\end{lemma}

\begin{proof}
  Suppose that $H$ has a representation where the hyperedges in $F$
  are all convex. Let $e=abcd$. (Note that we identify~$e$ with the
  quadrilateral that represents it.)  No matter which segments form
  the diagonals of $e$
  ($ab$ and~$cd$, $ac$ and~$bd$, or~$ad$ and~$bc$), there is a pair
  ($uv$, $wx$, or $yz$) that forms two hyperedges with the two
  diagonals.
  We assume that the diagonals are~$ab$ and~$cd$ forming hyperedges
  $abuv$ and $cduv$.  Due to the convexity
  of~the quadrilateral~$e$, $ab$ and~$cd$ intersect.  Hence, since
  $abuv$ and $cduv$ are convex, one of
  the hyperedges~$abuv$ and~$cduv$ must be drawn above~$e$, and 
  the other below~$e$.
  This yields the desired contradiction since~$u$ and~$v$
  are contained in both of these hyperedges.

  Suppose that $H$ has a representation where the hyperedges in $F$
  are all non-convex.  We may
  assume that the diagonals of $e$ are again~$ab$ and~$cd$, that~$ab$ is not
  contained in~$e$, and that $c$ lies on the convex hull of~$e$
  whereas~$d$ does not.  Let~$x$ be the intersection point of the
  supporting lines of~$ab$ and~$cd$.  Note that~$x$ lies on~$ab$.
  This is due to the fact that~$cd$ is incident to~$c$ and lies inside
  the angle $\angle acb$.  The supporting planes of the
  quadrilaterals~$abuv$ and~$cduv$ intersect in a line~$\ell$ that
  intersects the supporting plane of~$e$ in~$x$.  Clearly, $u$ and~$v$
  must lie on~$\ell$ since they are part of both quadrilaterals~$abuv$
  and~$cduv$.

  We consider two cases.  In the first case, $u$ and~$v$ lie on
  different halflines of~$\ell$ with respect to~$x$.  Then
  vertices~$a$, $b$, $u$, and~$v$ are in convex position, forming a
  convex quadrilateral with diagonals~$uv$ and~$ab$.  Note that it is
  not possible to connect points in convex position with straight-line
  segments to form a non-degenerate non-convex polygon.  This
  contradicts the fact that all quadrilaterals in $F$ must be non-convex.  In
  the second case, $u$ and~$v$ lie on the same halfline of~$\ell$ with
  respect to~$x$.  But then vertices~$c$, $d$, $u$, and~$v$ are in
  convex position since also~$c$ and~$d$ lie on the same halfline with
  respect to~$x$.  This again yields the desired contradiction.
\end{proof}

\begin{proposition} \label{thm:sqs10}
  The Steiner quadruple system $S(3,4,10)$ does
  not admit a non-crossing drawing in 3D, where all quadrilaterals representing the hyperedges are
  convex or all quadrilaterals are non-convex.
\end{proposition}

\begin{proof}
  The Steiner quadruple system $S(3,4,10)$ has ten vertices and $30$
  hyperedges and is unique; see Table~\ref{tab:SteinerSys}.
  Note that $S(3,4,10)$ satisfies the assumptions of  Lemma~\ref{lem:quadruple}
  for $a=1, b=4,c=2,d=5,u=7,v=9,w=6,x=0,y=3$, and $z=8$, i.e.,
  it contains the set of edges $F = \{1245,1260,4560,1479,2579,1538, 2438 \}$.
  The claim follows from Lemma~\ref{lem:quadruple}.
\end{proof}

\begin{theorem}
  \label{thm:no-sqs}
  No Steiner quadruple system admits a non-crossing drawing using
  convex quadrilaterals in~3D.
  If the system contains at least $20$ vertices, it does not admit a
  non-crossing drawing using any quadrilaterals in~3D.
\end{theorem}

\begin{proof}
  Day and Edelsbrunner~\cite[Lemma 2.3]{de-ctchp-DCG94} used an
  approach similar to that of Carmesin (mentioned in
  footnote~\ref{ftn:carmesin}) to show that the number of triangles
  spanned by $n$ points in 3D is less than $n^2$ if no two triangles
  have a non-trivial intersection.  (A trivial intersection is a
  common point or edge.)  We need to redo their proof taking lower-order terms
  into account.  If a Steiner quadruple system $S(3,4,n)$ can be drawn
  using quadrilaterals in 3D, the intersection of these quadrilaterals
  with a small sphere around a vertex is a planar graph.  Recall that
  any $S(3,4,n)$ has $n$ vertices and $m=\binom{n}{3}/4$ quadruples.
  Let $v$ be any vertex.  Then~$v$ is incident to $4m/n=(n-1)(n-2)/6$
  quadrilaterals.
Suppose that there is a representation consisting of only convex quadrilaterals. Break each convex quadrilateral incident to $v$ into two triangles such that both triangles are incident to $v$.
The intersection of these triangles with a small sphere around $v$ yields a graph on $n-1$ vertices (that is, on all vertices but~$v$) with $(n-1)(n-2)/3$ edges. For $n>8$, this graph cannot be planar. This, together with Proposition~\ref{thm:noS348asQuads}, yields the first part of our claim.

The same approach proves the second part as well.
Suppose that there is a representation without any restrictions.
For a non-convex quadrilateral incident to $v$, it may or may not be possible to break it into two triangles such that both are incident to $v$. Here, we can only break the quadrilaterals (convex or non-convex) for which this splitting is possible.
After this step, the polygons incident to $v$ are either triangles or quadrilaterals. The intersection of these polygons with a small sphere around $v$ yields a graph that has at most $n-1$ vertices and at least $(n-1)(n-2)/6$ edges. Such a graph cannot be planar for $n \ge 18$. Since the first Steiner quadruple system with $n \ge 18$ vertices has $20$ vertices, the proof is complete.
\end{proof}

\subsection{Conclusion and Open Problems}
\label{app:discussion}

We conclude the paper by pointing out some possible directions for
further research.

\paragraph{Representing Graphs.}

Our general construction in Section~\ref{sub:general} implies that
every $n$-vertex graph of maximum degree~$\Delta$ can be represented
in~3D by monotone polygonal curves with at most $\Delta-2$ bends.  Is
there a natural class of graphs such that every graph~$G$ in that
class can be represented by such chains with fewer than $\Delta(G)-2$
bends?

\begin{open}
  Does some non-trivial class of graphs admit an intersection
  representation in~3D using polygonal chains with less than
  $\Delta(G)-2$ bends for any graph~$G$ in that class?
\end{open}

Furthermore, it is interesting to consider representations of specific graph with highly regular polygons.
In particular, we suggest the following extension of the result in
Theorem~\ref{thm:squares-of-cycle}.

\begin{open}
If $k \geq 5$ is odd, can $C_k^2$ be represented by touching unit squares?
\end{open}

\paragraph{Steiner Triple Systems.}  We have constructed non-crossing
drawings of the two smallest such systems, $S(2,3,7)$ and $S(2,3,9)$,
using triangles; see Proposition~\ref{prop:fano}.  What about larger
Steiner triple systems?

\begin{open}
  Does any Steiner triple system $S(2,3,n)$ with $n \ge 13$ admit a
  non-crossing drawing using triangles in 3D?
\end{open}

\paragraph{Steiner Quadruple Systems.}  We have shown that no Steiner
quadruple system admits a crossing-free drawing using convex
quadrilaterals and that Steiner systems with exactly 8 or with at
least 20 vertices do not admit crossing-free drawings using arbitrary
quadrilaterals; see Theorem~\ref{thm:no-sqs} and
Proposition~\ref{thm:noS348asQuads}.

\begin{open}
  Does any Steiner quadruple system with $n \in \{10,14,16\}$ vertices
  admit a non-crossing drawing using ``mixed'' quadrilaterals, that
  is, some convex, some not?
\end{open}

\paragraph{Projective Planes.}

Note that in Steiner quadruple systems many pairs of edges intersect
in two vertices.  In a projective plane, every pair of edges (called
\emph{lines}) intersects in exactly one vertex (\emph{point}).  So
maybe this is easier?  Recall that any projective plane fulfills the
following axioms:
\begin{enumerate}[(P1)]
\item Given any two distinct points, there is exactly one line
  incident to both of them. \label{enum:point-line}
\item Given any two distinct lines, there is exactly one point
  incident to both of them. \label{enum:line-point}
\item There are four points such that no line is incident to more than
  two of them. (Non-degeneracy axiom) \label{enum:non-degeneracy}
\end{enumerate}
Every projective plane has the same number of lines as it has points.
The projective plane of order $N$, $PG(N)$, has $N^2+N+1$ lines and
points and there are $N+1$ points on each line, and $N+1$ lines go
through each point.  Equivalently, we can see $PG(N)$ as the Steiner system
$S(2,N+1,N^2+N+1)$.

\begin{figure}
  \hfill
  {\small
  \begin{tabular}[c]{cccc}
    \toprule
    \multicolumn{4}{c}{$PG(3)$}\\
    \midrule
    A & B & C & D\\
    A & 1 & 2 & 3\\
    A & 4 & 5 & 6\\
    A & 7 & 8 & 9\\
    B & 1 & 4 & 7\\
    B & 2 & 5 & 8\\
    B & 3 & 6 & 9\\
    C & 1 & 5 & 9\\
    C & 2 & 6 & 7\\
    C & 3 & 4 & 8\\
    D & 1 & 6 & 8\\
    D & 2 & 4 & 9\\
    D & 3 & 5 & 7\\
    \bottomrule
  \end{tabular}
  }
  \hfill
  \begin{minipage}[c]{.5\linewidth}
    \includegraphics{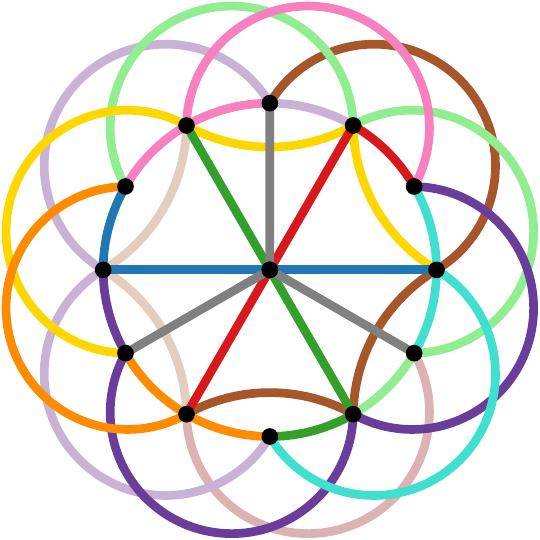}
  \end{minipage}
  \hfill\null

  \caption{The second smallest discrete projective plane $PG(3)$,
    which is a 4-regular 4-uniform hypergraph with 13 vertices and 13
    hyperedges.  The drawing was inspired by
    \url{https://puzzlewocky.com/games/the-math-of-spot-it/}.}
  \label{fig:projective}
\end{figure}

Note, however, that any contact representation of $PG(3)$ by convex
quadrilaterals contains a contact representation of $S(2,3,9)$ by
triangles: just drop one of the 13 quadruples of $PG(3)$ and remove
its four vertices from all quadruples.  This yields twelve triples
with the property that any pair of vertices is contained in a unique
triple~(P\ref{enum:point-line}).

\begin{observation}
  \label{obs:PG3}
  Suppose that there is a realization of $PG(3)$ with convex
  quadrilaterals, then no two quadrilaterals are coplanar in such a
  realization.
\end{observation}

\begin{proof}
  For the sake of contradiction, suppose that two quadrilaterals,
  $q_1$ and~$q_2$, lie in the same plane~$\Pi$ in a realization~$R$ of
  $PG(3)$ with convex quadrilaterals.  Every two quadrilaterals share
  exactly one vertex~(P\ref{enum:line-point}), hence, we can write
  $q_1$ and $q_2$ as $q_1=u_1u_2u_3w$ and $q_2=v_1v_2v_3w$.  Since
  every pair of vertices appears in exactly one
  quadrilateral~(P\ref{enum:point-line}), each pair $u_iv_j$ with $i,j
  \in \{1,2,3\}$ is contained in a different quadrilateral.  Since the
  quadrilaterals in~$R$ are convex, each line segment
  $\overline{u_iv_j}$ is contained in the unique quadrilateral
  containing vertices~$u_i$ and~$v_j$.  Since~$u_i$ and~$v_j$ lie
  in~$\Pi$, $\overline{u_iv_j}$ also lies in~$\Pi$.  As a result,
  $\Pi$ contains a planar (straight-line) drawing of~$K_{3,3}$ (with
  vertex set $\{u_1,u_2,u_3\} \cup \{v_1,v_2,v_3\}$), which yields the
  desired contradiction.
\end{proof}

Or is it perhaps more natural to represent $PG(3)$ by touching
tetrahedra?  We used an integer linear program to compute the poset
dimension of the 
vertex--hyperedge inclusion order of $PG(3)$, which turned out to
be~5.  Hence the method of Ossona de Mendez~\cite{o-rp-JGAA02} yields
a contact representation of $PG(3)$ by touching tetrahedra~-- but only
in~4D.

\section*{Acknowledgments.}

We are grateful to the organizers of the workshop Homonolo 2017, where
the project originates.  We thank G\"{u}nter Rote for advice regarding
strictly convex drawings of polygons on the grid, and we thank Torsten
Ueckerdt for bringing Ossona de Mendez' work~\cite{o-rp-JGAA02} to our
attention.  We are indebted to
Arnaud de Mesmay and Eric Sedgwick for pointing us to the lemma of
Dey and Edelsbrunner~\cite{de-ctchp-DCG94}, which yielded
Theorem~\ref{thm:no-sqs}.  We also thank Alexander Ravsky who pointed
us to the fact that the method of Ossona de Mendez needs~4D for
representing~$K_{13}$ as a 2-uniform hypergraph (with touching line
segments).  Last but not least we thank Oksana Firman for implementing
the above-mentioned integer linear program.

\bibliographystyle{abbrvurl}
\bibliography{abbrv,hyper,contact}
\end{document}

%% file: figures/bipar-tp.tex
\tdplotsetmaincoords{0}{0}
\begin{tikzpicture}[tdplot_main_coords,scale=.7,blend group=normal,tdplot_main_coords]

\def \s {1.5};

\def \r {180/7};

\foreach \i in {0,1,...,7}{
	\coordinate (P0\i) at ({-cos(-90+\r*\i )}, 0, {sin(-90+\r*\i )});
	\coordinate(O\i) at (-2,0,{sin(-90+\r*\i )});
}

\foreach \i in {1,2,...,7}
	\foreach \j in {0,1,...,7}
		\coordinate (P\i\j) at ($(O\j)!1!45*\i:(P0\j)$);

\foreach \i in {0,1,6,7}{
	\draw[fill,darkgreen,opacity=.2] (P0\i) -- (P1\i) -- (P2\i) -- (P3\i) -- (P4\i) -- (P5\i) -- (P6\i) -- (P7\i) -- cycle;
	}
\foreach \i in {2,5}{
	\draw[fill,darkgreen,opacity=.3] (P0\i) -- (P1\i) -- (P2\i) -- (P3\i) -- (P4\i) -- (P5\i) -- (P6\i) -- (P7\i) --  cycle;
	}
\foreach \i in {3,4}{
	\draw[fill,darkgreen,opacity=.7] (P0\i) -- (P1\i) -- (P2\i) -- (P3\i) -- (P4\i) -- (P5\i) -- (P6\i) -- (P7\i) --  cycle;
	}

\foreach \i in{0,1,2,3,4,5,6,7}
	\draw[fill,blue,opacity=.5,very thick] (P\i0) --(P\i1) -- (P\i2) -- (P\i3) -- (P\i4) -- (P\i5) --(P\i6) -- (P\i7) --cycle;

\foreach \i in {0,1,...,7}{
	\foreach \j in {0,1,...,7}
		\draw[fill,opacity=.8] (P\i\j) circle (\s pt);
}

\end{tikzpicture}

%% file: figures/bipar-t.tex
\tdplotsetmaincoords{82}{55}
\begin{tikzpicture}[tdplot_main_coords,scale=.75,blend group=normal, xscale=2]

\def \s {1};

\def \r {180/7};

\foreach \i in {0,1,...,7}{
	\coordinate (P0\i) at ({-cos(-90+\r*\i )}, 0, {sin(-90+\r*\i )});
	\coordinate(O\i) at (-2,0,{sin(-90+\r*\i )});
}

\foreach \i in {1,2,...,7}
	\foreach \j in {0,1,...,7}
		\coordinate (P\i\j) at ($(O\j)!1!45*\i:(P0\j)$);

\foreach \i in {0,1,2,3}{
	\draw[fill,darkgreen,opacity=.3,draw=darkgreen!70!black,very thick] (P0\i) -- (P1\i) -- (P2\i) -- (P3\i) -- (P4\i) -- (P5\i) -- (P6\i) -- (P7\i) -- cycle;
	}

\foreach \i in{1,2,3,4}
	\draw[fill,blue,opacity=.5,very thick] (P\i0) --(P\i1) -- (P\i2) -- (P\i3) -- (P\i4) -- (P\i5) --(P\i6) -- (P\i7) --cycle;

\foreach \i in {0,1,...,7}{
	\foreach \j in {0,1,...,7}
		\node[circle,fill,opacity=.8,minimum size=1mm,inner sep=1pt] at (P\i\j) {};
}
 
	\foreach \i in {4,5,6}{
	\draw[fill,darkgreen,opacity=.5,draw=darkgreen!70!black,very thick] (P0\i) -- (P1\i) -- (P2\i) -- (P3\i) -- (P4\i) -- (P5\i) -- (P6\i) -- (P7\i) --  cycle;
	}

\foreach \i in{0,5,6,7}
	\draw[fill,blue,opacity=.65,very thick] (P\i0) --(P\i1) -- (P\i2) -- (P\i3) -- (P\i4) -- (P\i5) --(P\i6) -- (P\i7) --cycle;

\foreach \i in{7} {	
	\draw[fill,darkgreen,opacity=.65,draw=darkgreen!70!black,very thick] (P0\i) -- (P1\i) -- (P2\i) -- (P3\i) -- (P4\i) -- (P5\i) -- (P6\i) -- (P7\i) --  cycle;
	}

\end{tikzpicture}

%% file: figures/bipar-ip2.tex
\tdplotsetmaincoords{0}{0}
\begin{tikzpicture}[tdplot_main_coords,scale=.35,blend group=normal]

\def \s {2};
\def \t {4};
\edef \c {0};

\coordinate (P0\t) at (0, 0, \t);
\coordinate (P1\t) at (1, 0, \t);
\coordinate (P2\t) at (4, 1, \t);
\coordinate (P3\t) at (6, 2, \t);
\coordinate (P4\t) at (7,3, \t);
\coordinate (P5\t) at (7,4, \t);
\coordinate (P6\t) at (6,5, \t);
\coordinate (P7\t) at (4, 6, \t);
\coordinate (P8\t) at (1, 7, \t);
\coordinate (P9\t) at (0, 7, \t);
\coordinate (P10\t) at (-3, 6, \t);
\coordinate (P11\t) at (-5, 5, \t);
\coordinate (P12\t) at (-6, 4, \t);
\coordinate (P13\t) at (-6, 3, \t);
\coordinate (P14\t) at (-5, 2, \t);
\coordinate (P15\t) at (-3, 1, \t);

\foreach \i in {\t,...,2} {
	\xdef \c {\c+1};
	\pgfmathparse {\i-1};
	\pgfmathdiv{\pgfmathresult}{1};
	\xdef \k {\pgfmathresult};

	\foreach \j in {1,...,8}
		\coordinate (P\j\k) at ($(P\j\i) + (\c,0,-1)$);

	\foreach \j in {9,...,15,0}
		\coordinate (P\j\k) at ($(P\j\i) - (\c,0,1)$);
	
}

\foreach \i in {5,6,7,8}{
	\pgfmathparse{9-\i}
	\pgfmathdiv{\pgfmathresult}{1};
	\xdef \k {\pgfmathresult};

	\foreach \j in {0,...,15}{
		\coordinate (P\j\i) at ($(P\j\k) + (0,0,\i-\k)$);
	 }

}

\foreach \i in {1,2,7,8}{
	\draw[fill,darkgreen,opacity=.2] (P0\i) -- (P1\i) -- (P2\i) -- (P3\i) -- (P4\i) -- (P5\i) -- (P6\i) -- (P7\i) -- (P8\i) -- (P9\i) -- (P10\i) -- (P11\i) --(P12\i) -- (P13\i) -- (P14\i) -- (P15\i) -- cycle;
}
\foreach \i in {3,6}{
	\draw[fill,darkgreen,opacity=.4] (P0\i) -- (P1\i) -- (P2\i) -- (P3\i) -- (P4\i) -- (P5\i) -- (P6\i) -- (P7\i) -- (P8\i) -- (P9\i) -- (P10\i) -- (P11\i) -- (P12\i) -- (P13\i) -- (P14\i) -- (P15\i) -- cycle;
}
\foreach \i in {4,5}{
	\draw[fill,darkgreen,opacity=.6] (P0\i) -- (P1\i) -- (P2\i) -- (P3\i) -- (P4\i) -- (P5\i) -- (P6\i) -- (P7\i) --  (P8\i) -- (P9\i) -- (P10\i) -- (P11\i) --(P12\i) -- (P13\i) -- (P14\i) -- (P15\i) --  cycle;
}

\foreach \i in {0,...,15}
	\draw[fill,blue,very thick] (P\i1) --(P\i2) -- (P\i3) -- (P\i4) -- (P\i5) -- (P\i6) --(P\i7) -- (P\i8) --cycle;

\foreach \i in {1,...,8}
	\foreach \j in {0,1,...,15}
		\draw[fill,opacity=.8] (P\j\i) circle (\s pt);

\draw[dotted,darkgreen] ($(P44)-(0,5,0)$) -- ($(P54)+(0,5,0)$);
\draw[dotted,blue] ($(P134)-(0,5,0)$) -- ($(P124)+(0,5,0)$);
\draw[dotted,blue] ($(P133)-(0,5,0)$) -- ($(P123)+(0,5,0)$);
\draw[dotted,blue] ($(P132)-(0,5,0)$) -- ($(P122)+(0,5,0)$);
\draw[dotted,blue] ($(P131)-(0,5,0)$) -- ($(P121)+(0,5,0)$);

\draw[dotted,darkgreen] ($(P34)-(0,3,0)$) -- ($(P64)+(0,3,0)$);
\draw[dotted,darkgreen] ($(P24)-(0,2,0)$) -- ($(P74)+(0,2,0)$);
\draw[dotted,darkgreen] ($(P14)-(0,1,0)$) -- ($(P84)+(0,1,0)$);

\draw[dotted,darkgreen] ($(P04)-(0,1,0)$) -- ($(P94)+(0,1,0)$);

\draw[dotted,darkgreen] ($(P154)-(0,2,0)$) -- ($(P104)+(0,2,0)$);
\draw[dotted,darkgreen] ($(P144)-(0,3,0)$) -- ($(P114)+(0,3,0)$);
\draw[dotted,darkgreen] ($(P134)-(0,4,0)$) -- ($(P124)+(0,4,0)$);

\draw[<->,blue] ($(P121)-(0,-5.5,0)$) -- node[above]{\scriptsize{$3$}}  ($(P122)-(0,-5.5,0)$);
\draw[<->,blue] ($(P122)-(0,-5.5,0)$) -- node[above]{\scriptsize{$2$}}  ($(P123)-(0,-5.5,0)$);
\draw[<->,blue] ($(P123)-(0,-5.5,0)$) -- node[above]{\scriptsize{$1$}}  ($(P124)-(0,-5.5,0)$);
\draw[<->] ($(P131)-(0,5.5,0)$) -- node[below] {\scriptsize{$\left\lceil \frac{|A|}{2} \right\rceil -1$}} ($(P132)-(0,5.5,0)$);

\draw[<->,darkgreen] ($(P124)-(0,-4.5,0)$) -- node[above]{\scriptsize{$1$}}  ($(P114)-(0,-3.5,0)$);
\draw[<->,darkgreen] ($(P114)-(0,-3.5,0)$) -- node[above]{\scriptsize{$2$}}  ($(P104)-(0,-2.5,0)$);
\draw[<->,darkgreen] ($(P104)-(0,-2.5,0)$) -- node[above]{\scriptsize{$3$}}  ($(P94)-(0,-1.5,0)$);
\draw[<->] ($(P154)-(0,2,0)$) -- node[below]{\scriptsize{$\left\lceil \frac{|B|}{4}\right\rceil-1$}}  ($(P04)-(0,1,0)$);

\draw[<->,darkgreen] ($(P94)-(0,-1.5,0)$) -- node[above]{\scriptsize{$1$}}  ($(P84)-(0,-1.5,0)$);

\draw[<->,darkgreen] ($(P84)-(0,-1.5,0)$) -- node[above]{\scriptsize{$3$}}  ($(P74)-(0,-2.5,0)$);
\draw[<->] ($(P24)-(0,3,0)$) -- node[below]{\scriptsize{$\left\lfloor \dfrac{\left\lceil \frac{|B|}{2}\right\rceil}{2} \right\rfloor-1$}}  ($(P14)-(0,2,0)$);

\draw[<->,darkgreen] ($(P74)-(0,-2.5,0)$) -- node[above]{\scriptsize{$2$}}  ($(P64)-(0,-3.5,0)$);
\draw[<->,darkgreen] ($(P64)-(0,-3.5,0)$) -- node[above]{\scriptsize{$1$}}  ($(P54)-(0,-4.5,0)$);

\draw[<->] ($(P124)+(0,6.5,0)$) -- node[above]{\scriptsize{$\left\lfloor \dfrac{\left\lceil\frac{|B|}{2}\right\rceil}{2} \right\rfloor \left(\left\lceil \frac{|B|}{4} \right\rceil -1\right) +1$}}  ($(P54)+(0,6.5,0)$);

\draw[<->] ($(P131) + (-1,-3,0)$) -- node[above,rotate=90] {\scriptsize{$2\left\lceil \frac{|B|}{4} \right\rceil -1$}} ($(P121) + (-1,3,0)$);

\node[below left, inner sep=-.5mm,xshift=-1mm,yshift=-1mm] at (0,3) {$z$};
\draw[->] (0,3) -++(14,0) node[right] {$x$};
\draw[->] (0,3) -++(0,5) node[left,yshift=-1mm] {$y$};

\end{tikzpicture}

%% file: figures/bipar-i.tex
\tdplotsetmaincoords{83}{55}
\begin{tikzpicture}[tdplot_main_coords,scale=.45,blend group=normal]

\def \s {1.7};
\def \t {4};
\edef \c {0};

\coordinate (P0\t) at (0, 0, \t);
\coordinate (P1\t) at (1, 0, \t);
\coordinate (P2\t) at (2, 1, \t);
\coordinate (P3\t) at (2, 2, \t);
\coordinate (P4\t) at (1,3, \t);
\coordinate (P5\t) at (0,3, \t);
\coordinate (P6\t) at (-1, 2, \t);
\coordinate (P7\t) at (-1, 1, \t);

\foreach \i in {\t,...,2} {
	\xdef \c {\c+1};
	\pgfmathparse {\i-1};
	\pgfmathdiv{\pgfmathresult}{1};
	\xdef \k {\pgfmathresult};

	\foreach \j in {1,...,\t}
		\coordinate (P\j\k) at ($(P\j\i) + (\c,0,-1)$);

	\foreach \j in {5,6,7,0}
		\coordinate (P\j\k) at ($(P\j\i) - (\c,0,1)$);
	
}

\foreach \i in {5,6,7,8}{
	\pgfmathparse{9-\i}
	\pgfmathdiv{\pgfmathresult}{1};
	\xdef \k {\pgfmathresult};

	\foreach \j in {0,...,7}{
		\coordinate (P\j\i) at ($(P\j\k) + (0,0,\i-\k)$);
	 }

}

\foreach \i in {5,6,7}
	\draw[fill,blue,opacity=.3,very thick] (P\i1) --(P\i2) -- (P\i3) -- (P\i4) -- (P\i5) -- (P\i6) --(P\i7) -- (P\i8) --cycle;

\foreach \i in {1,...,8}{
	\draw[fill,darkgreen,opacity=.7] (P0\i) -- (P1\i) -- (P2\i) -- (P3\i) -- (P4\i) -- (P5\i) -- (P6\i) -- (P7\i) -- cycle;
}

\foreach \i in {4,3}
	\draw[fill,blue!80,opacity=.3,very thick] (P\i1) --(P\i2) -- (P\i3) -- (P\i4) -- (P\i5) -- (P\i6) --(P\i7) -- (P\i8) --cycle;
\foreach \i in {2,1}
	\draw[fill,blue,opacity=.4,very thick] (P\i1) --(P\i2) -- (P\i3) -- (P\i4) -- (P\i5) -- (P\i6) --(P\i7) -- (P\i8) --cycle;

\foreach \i in {0}
	\draw[fill,blue,opacity=.4,very thick] (P\i1) --(P\i2) -- (P\i3) -- (P\i4) -- (P\i5) -- (P\i6) --(P\i7) -- (P\i8) --cycle;

\foreach \i in {1,...,8}
	\foreach \j in {0,1,...,7}
		\draw[fill,opacity=.8] (P\j\i) circle (\s pt);

\draw[->] (0,3,0) -++(0,0,9) node[above] {$z$};
\end{tikzpicture}

%% file: figures/K33.tex
\begin{tikzpicture}[x  = {(5.5mm,2.5mm)},
                    y={(-1.1cm,0.4cm)},x={(1cm,0.5cm)},
                    z={(0mm,9mm)},
                   scale=2.5]

\coordinate (P1) at (1,0,0);
\coordinate (P2) at (.5,0.87,0);
\coordinate (P3) at (0,0,0);
\coordinate (P4) at (1.05, 0.19, 0.5);
\coordinate (P5) at (0.3,0.81,0.5);
\coordinate (P6) at (0.15,-0.15,0.5);
\coordinate (P7) at (1,0,1);
\coordinate (P8) at (.5,0.87,1);
\coordinate (P9) at (0,0,1);

\draw[fill,blue,opacity=.3] (P2) -- (P4) --(P8) -- (P2);

\draw[fill,green,opacity=.3] (P1) -- (P2) --(P3) -- (P1);
\draw[line width=1pt] (P1) -- (P2) --(P3) -- (P1);
\draw[fill,green,opacity=0.3] (P4) -- (P5) --(P6) -- (P4);
\draw[dashed,line width=1pt] (P4) -- (P5) --(P6) -- (P4);
\draw[fill,green,opacity=0.3,] (P7) -- (P8) --(P9) -- (P7);
\draw[line width=1pt] (P7) -- (P8) --(P9) -- (P7);

\foreach \i in {4,5,6} {
  \draw[fill] (P\i) circle (1.2pt);
  \draw[fill,white] (P\i) circle (.9pt);
}

\draw[fill,blue,opacity=.3] (P1) -- (P6) --(P7) -- (P1);
\draw[fill,blue,opacity=.3] (P3) -- (P5) --(P9) -- (P3);

\foreach \i in {1,2,3,7,8,9}
  \draw[fill] (P\i) circle (1.2pt);

\end{tikzpicture}

%% file: figures/K33C.tex
\begin{tikzpicture}[scale=2.2]
 
	\pgfmathsetmacro \r {sin(60)};
	\pgfmathsetmacro \R {tan(30)};
	\pgfmathsetmacro \a {acos(-1/8)};
	\pgfmathsetmacro \b {-acos(-1/8)};

	\coordinate (P1) at (0,0);
	\coordinate (P2) at (1,0);
	\coordinate (X) at ($(P1)!.5!(P2)$);
	\coordinate (P3) at ($(X)!sin(60)*2!90:(P2)$);
	\coordinate (O) at ($(X)!1/3!(P3)$);

	\draw[gray!50] (O) circle (\R);

	\foreach \i in {1,...,3}{
          \coordinate (G\i) at 	($(O)!1!\a:(P\i)$);	
          \coordinate (H\i) at 	($(O)!1!\b:(P\i)$);	
          \draw[line width=2.75 pt,cap=round,blue,opacity=.5] (P\i) -- (G\i);
        }
        \draw[gray!75] (P1) circle (\r);
        \draw[gray!75] (P2) circle (\r);
        \draw[gray!75] ([shift=(170:\r)]P3) arc (170:370:\r);

        \draw[fill,green,opacity=.3] (P1) -- (P2) --(P3) -- (P1);
	\draw[line width=1pt] (P1) -- (P2) --(P3) --(P1);

	\draw[fill,green,opacity=.3] (G1) -- (G2) -- (G3) -- (G1);	
	\draw[dashed,line width=1pt] (G1) -- (G2) -- (G3) -- (G1);

	\draw[black,thick,dotted] (P1) -- (O) -- (G1) -- (P1);
        
        \draw[gray] (O) -- (G3);
        \coordinate[label=below:{\small{$\beta$}}] (B) at ($(O)!.4!(G3) +(-.05,.055)$); 
	\coordinate[label=below:{$\alpha$}] (A) at ($(X)!1/3!(P3)-(.015,.01)$);

        \draw[fill,white] (O) circle  (1pt);
        \draw[] (O) circle  (1pt);

	\foreach \i in {1,...,3}{
          \draw[] (H\i) circle  (.75pt);
          \draw[fill] (P\i) circle  (1.2pt);
          \draw[fill] (G\i) circle  (1.2pt);
          \draw[white,fill] (G\i) circle  (0.9pt);
	}
\end{tikzpicture}

%% file: figures/c2_6.tex
\tdplotsetmaincoords{75}{10}
\begin{tikzpicture}[tdplot_main_coords,scale = 1.6,blend group=darken,scale=0.9]
\def \s {.7};

\coordinate (P16) at (1,0,0);
\coordinate (P56) at (0.5,{sqrt(3)/2},0);
\coordinate (P45) at (-0.5,{sqrt(3)/2},0);
\coordinate (P34) at (-1,0,0);
\coordinate (P23) at (-0.5,{-sqrt(3)/2},0);
\coordinate (P12) at (0.5,{-sqrt(3)/2},0);

\coordinate  (P26) at (0.5,-{sqrt(3)/6},-.71) ;
\coordinate  (P46) at (0,{sqrt(3)/3},-.71);
\coordinate  (P24) at (-0.5,{-sqrt(3)/6},-.71);

\coordinate  (P15) at (0.5,{sqrt(3)/6},.71) ;
\coordinate  (P13) at (0,{-sqrt(3)/3},.71);
\coordinate  (P35) at (-0.5,{sqrt(3)/6},.71);

\foreach \i in {1,2,3,4,5}{
	\pgfmathparse{\i+1};
	\draw[fill,opacity=1] (P\i\pgfmathresult) circle (\s pt);
}
\draw[fill,opacity=1] (P16) circle (\s pt);

\foreach \i in {1,2,3,4}{
	\pgfmathparse{\i+2};
	\draw[fill,opacity=1] (P\i\pgfmathresult) circle (\s pt);
}

\draw[pink,fill=pink,opacity=.5] (P15) -- (P35) -- (P45) -- (P56) -- cycle;
\draw[yl,fill=yl,opacity=.35] (P34) -- (P45) -- (P46) -- (P24) -- cycle;
\draw[cyan,fill=cyan,opacity=.5] (P16) -- (P56) -- (P46) -- (P26) -- cycle;
\draw[red,line width=1.5pt,semitransparent,opacity=.5] (P15) -- (P26);
\draw[blue,fill=blue,opacity=.65] (P12) -- (P16) -- (P15) -- (P13) -- cycle;
\draw[orange,fill=orange,opacity=.6] (P35) -- (P13) -- (P23) -- (P34) -- cycle;
\draw[darkgreen,fill=darkgreen,opacity=.6] (P24) -- (P26) -- (P12) -- (P23) -- cycle;

\draw[fill,red,opacity=1] (P15) circle (1.8*\s pt);
\draw[fill,red,opacity=1] (P26) circle (1.8*\s pt);

\node[label=below:{$\textcolor{darkgreen!70!black}{Q_2}$},inner sep=0pt,xshift=-.5mm,yshift=0.5mm] at ($(P24)!.5!(P26)$) {};
\node[label=left:{$\textcolor{orange!70!black}{Q_3}$},inner sep=0pt] at ($(P35)!.5!(P34)$) {};
\node[label=above:{$\textcolor{pink!60!black}{Q_5}$},inner sep=0pt,yshift=-1mm] at ($(P35)!.5!(P15)$) {};
\node[label=right:{$\textcolor{blue!70!black}{Q_1}$},inner sep=0pt,xshift=-1mm,yshift=1mm] at ($(P15)!.5!(P16)$) {};
\node[label=right:{$\textcolor{cyan!70!black}{Q_6}$},inner sep=0pt,xshift=-1mm,yshift=-1.5mm] at ($(P16)!.5!(P26)$) {};
\node[label=left:{$\textcolor{yl!70!black}{Q_4}$},inner sep=0pt,xshift=1.5mm,yshift=-1mm] at ($(P34)!.5!(P24)$) {};

\node[label=right:{$\color{red}{T}$},inner sep=0pt,yshift=2mm,xshift=0mm] at (P15) {};
\node[label=right:{$\color{red}{B}$},inner sep=0pt,yshift=-2mm,xshift=-1mm] at (P26) {};


\end{tikzpicture}

%% file: figures/c2_7-s.tex
\tdplotsetmaincoords{75}{30}
\begin{tikzpicture}[tdplot_main_coords,scale = 1.5,blend group=darken]
\def \s {.7};

\coordinate (L1) at (1, -.38, -.24);
\coordinate (L2) at (1, 0.38, 0.24);

\coordinate (P16) at (1,0,0);
\coordinate (P56) at (0.5,{sqrt(3)/2},0);
\coordinate (P45) at (-0.5,{sqrt(3)/2},0);
\coordinate (P34) at (-1,0,0);
\coordinate (P23) at (-0.5,{-sqrt(3)/2},0);
\coordinate (P12) at (0.5,{-sqrt(3)/2},0);

\coordinate  (P27) at (0.5,-{sqrt(3)/6},-.71);
\coordinate  (P46) at (0,{sqrt(3)/3},-.71);
\coordinate  (P24) at (-0.5,{-sqrt(3)/6},-.71);

\coordinate  (P57) at (0.5,{sqrt(3)/6},.71);
\coordinate  (P13) at (0,{-sqrt(3)/3},.71);
\coordinate  (P35) at (-0.5,{sqrt(3)/6},.71);

\coordinate (P17) at ($(P57)!.4!(P12)$);
\coordinate (P67) at ($(P27)!.4!(P56)$);

\foreach \i in {1,2,3,4,5}{
	\pgfmathparse{\i+1};
	\draw[fill,opacity=1] (P\i\pgfmathresult) circle (\s pt);
}

\foreach \i in {1,2,3,4}{
	\pgfmathparse{\i+2};
	\draw[fill,opacity=1] (P\i\pgfmathresult) circle (\s pt);
}

\draw[fill,opacity=1] (P16) circle (\s pt);
\foreach \i in {1,2,5,6}
	\draw[fill,opacity=1,red] (P\i7) circle (\s pt);

\draw[pink,fill=pink,opacity=.5] (P57) -- (P35) -- (P45) -- (P56) -- cycle;
\draw[red,fill=red,opacity=.4] (P17) -- (P57) -- (P67) -- (P27) -- cycle;
\draw[blue,fill=blue,opacity=.65] (P12) -- (P16) -- (P17) -- (P13) -- cycle;
\draw[orange,fill=orange,opacity=.65] (P35) -- (P13) -- (P23) -- (P34) -- cycle;
\draw[yl,fill=yl,opacity=.35] (P34) -- (P45) -- (P46) -- (P24) -- cycle;
\draw[cyan,fill=cyan,opacity=.35] (P16) -- (P56) -- (P46) -- (P67) -- cycle;
\draw[darkgreen,fill=darkgreen,opacity=.65] (P24) -- (P27) -- (P12) -- (P23) -- cycle;

\node[label=right:{$\color{red}{B=P_{2,n+1}}$},inner sep=0pt,xshift=-5mm,yshift=-2mm] at (P27) {};
\node[label=right:{$\color{red}{T=P_{n-1,n+1}}$},inner sep=0pt,xshift=-5mm,yshift=2mm] at (P57) {};
\node[label=right:{$\color{red}{P_{n,n+1}}$},inner sep=0pt,yshift=0mm,xshift=-.5mm] at (P67) {};
\node[label=right:{$\color{red}{P_{1,n+1}}$},inner sep=0pt,yshift=0mm,xshift=0mm] at (P17) {};


\end{tikzpicture}

%% file: figures/c2_7-u.tex
\tdplotsetmaincoords{10}{100}
\begin{tikzpicture}[tdplot_main_coords,scale = 1.5,blend group=darken,scale=0.9]
\def \s {.7};

\coordinate (P16) at (1,0,0);
\coordinate (P56) at (0.5,{sqrt(3)/2},0);
\coordinate (P45) at (-0.5,{sqrt(3)/2},0);
\coordinate (P34) at (-1,0,0);
\coordinate (P23) at (-0.5,{-sqrt(3)/2},0);
\coordinate (P12) at (0.5,{-sqrt(3)/2},0);

\coordinate  (P27) at (0.5,-{sqrt(3)/6},-.71);
\coordinate  (P46) at (0,{sqrt(3)/3},-.71);
\coordinate  (P24) at (-0.5,{-sqrt(3)/6},-.71);

\coordinate  (P57) at (0.5,{sqrt(3)/6},.71);
\coordinate  (P13) at (0,{-sqrt(3)/3},.71);
\coordinate  (P35) at (-0.5,{sqrt(3)/6},.71);

\coordinate (P17) at ($(P57)!.4!(P12)$);
\coordinate (P67) at ($(P27)!.4!(P56)$);

\foreach \i in {1,2,3,4,5}{
	\pgfmathparse{\i+1};
	\draw[fill,opacity=1] (P\i\pgfmathresult) circle (\s pt);
}

\foreach \i in {1,2,3,4}{
	\pgfmathparse{\i+2};
	\draw[fill,opacity=1] (P\i\pgfmathresult) circle (\s pt);
}

\draw[fill,opacity=1] (P16) circle (\s pt);
\foreach \i in {1,2,5,6}
	\draw[fill,opacity=1,red] (P\i7) circle (\s pt);

\draw[darkgreen,fill=darkgreen,opacity=.35] (P24) -- (P27) -- (P12) -- (P23) -- cycle;
\draw[yl,fill=yl,opacity=.35] (P34) -- (P45) -- (P46) -- (P24) -- cycle;
\draw[cyan,fill=cyan,opacity=.35] (P16) -- (P56) -- (P46) -- (P67) -- cycle;
\draw[red,fill=red,opacity=.45] (P17) -- (P57) -- (P67) -- (P27) -- cycle;
\draw[pink,fill=pink,opacity=.9] (P57) -- (P35) -- (P45) -- (P56) -- cycle;
\draw[blue,fill=blue,opacity=.65] (P12) -- (P16) -- (P17) -- (P13) -- cycle;
\draw[orange,fill=orange,opacity=.7] (P35) -- (P13) -- (P23) -- (P34) -- cycle;

\end{tikzpicture}

%% file: figures/fano-2d.tex
\begin{tikzpicture}[scale=2.75]
 
	\pgfmathsetmacro \r {sin(60)};
	\pgfmathsetmacro \R {tan(30)};

	\coordinate[label=below left:{5}] (P5) at (0,0);
	\coordinate[label=below right:{7}] (P3) at (1,0) ;
	\coordinate [label=below:{2}] (P4) at ($(P5)!.5!(P3)$) ;
	\coordinate [label=above:{3}] (P1) at ($(P4)!sin(60)*2!90:(P3)$) ;
	\coordinate [label=left:{1~~}] (P7) at ($(P4)!1/3!(P1)$);
	\coordinate [label=left:{4}] (P6) at ($(P1)!.5!(P5)$) ;
	\coordinate [label=right:{6}] (P2) at ($(P1)!.5!(P3)$);

	\draw (P7) circle (\R/2);

	\draw (P1) -- (P2) --(P3) -- (P4) -- (P5) -- (P6) --(P1);
	\draw (P1) -- (P7) -- (P4);
	\draw (P5) -- (P7) -- (P2);
	\draw (P3) -- (P7) -- (P6); 

	\foreach \i in {1,...,7}
		\draw[fill] (P\i) circle (.7pt);
		
	\node at (P7) {};

\end{tikzpicture}

%% file: figures/fano-t2.tex
\begin{tikzpicture}[scale=3,blend group=normal,rotate=60]
\def \s {.8};
\def \lww {.8};

\definecolor{ci}{RGB}{46,189,78}
\definecolor{co}{RGB}{11,83,69}

\coordinate (P1) at (0,0);
\coordinate (P2) at ({-sin(30)},{-tan(30)/2});
\coordinate (P6) at ($(P1)!1!120:(P2)$);
\coordinate (P4) at ($(P1)!1!-120:(P2)$);
\coordinate (P7) at ($(P1)!1!85:(P2)$);
\coordinate (P3) at ($(P1)!1!85:(P6)$);
\coordinate (P5) at ($(P1)!1!85:(P4)$);
  
\foreach \i in {4,6}
  \draw[fill] (P\i) circle (\s pt);
  
\path[fill=blue!60,opacity=.5,name path=b0] (P4) -- (P2) -- (P6) -- (P4);
\path[fill=ci!80,opacity=.4,name path=g3,semitransparent] (P1) -- (P2) -- (P3) -- (P1);    
\path[fill=ci!80,opacity=.4,name path=g2,semitransparent] (P1) -- (P4) -- (P7) -- (P1);    
\path[fill=ci!80,opacity=.4,name path=g1,semitransparent] (P1) -- (P5) -- (P6) -- (P1);     
\path[draw,blue,fill,opacity=.5,name path=b2] (P4) -- (P5) -- (P3) -- (P4);
\path[draw,blue,fill,opacity=.5,name path=b3] (P2) -- (P5) -- (P7) -- (P2);  
\path[draw,blue,fill,opacity=.5,name path=b1] (P6) -- (P7) -- (P3) -- (P6);   
\draw[line width=\lww pt,blue] (P4) -- (P5) -- (P3) -- (P4);
\draw[line width=\lww pt,blue] (P2) -- (P5) -- (P7) -- (P2);  
\draw[line width=\lww pt,blue] (P6) -- (P7) -- (P3) -- (P6);   

\node[label=left:{$1$},inner sep=1pt] {};
\node[label=below:{$2$},inner sep=1pt] at (P2) {};
\node[label=right:{$6$},inner sep=1pt] at (P6) {};
\node[label=left:{$4$},inner sep=1pt] at (P4) {};
\node[label=below:{$7$},inner sep=1pt] at (P7) {};
\node[label=above:{$3$},inner sep=.8pt] at (P3) {};
\node[label=left:{$5$},inner sep=1pt] at (P5) {};

{
blend group=multiply;
\foreach \k in {1,2,3}{

	\path [name intersections={of={b\k} and b0}];

	\foreach \i in {2,3}{
		\coordinate (C\k\i) at (intersection-\i);
		\ifnum \k=1
			\draw[line width=\lww pt,blue,dotted] (P6) -- (C\k\i);
		\fi
		\ifnum \k=2
			\draw[line width=\lww pt,blue,dotted] (P4) -- (C\k\i);
		\fi
		\ifnum \k=3
			\draw[line width=\lww pt,blue,dotted] (P2) -- (C\k\i);
		\fi
	}

	\path [name intersections={of={b\k} and {g\k}}];
	\foreach \i in {2,3}{
		\coordinate (D\k\i) at (intersection-\i);
		\ifnum \k=1
			\draw[line width=\lww pt,co,dotted] (P6) -- (D\k\i);
		\fi
		\ifnum \k=2
			\draw[line width=\lww pt,co,dotted] (P4) -- (D\k\i);
		\fi
		\ifnum \k=3
			\draw[line width=\lww pt,co,dotted] (P2) -- (D\k\i);
		\fi			
	}
}

	\path [name intersections={of={g1} and b0}];
	\draw[line width=\lww pt,blue,dotted] (intersection-1) -- (intersection-2);
	\draw[line width=\lww pt,blue] (intersection-1) -- (C22);	
	\draw[line width=\lww pt,blue] (intersection-2) -- (C32);			

	\path [name intersections={of={g2} and b0}];
	\draw[line width=\lww pt,blue,dotted] (intersection-3) -- (intersection-2);
	\draw[line width=\lww pt,blue] (intersection-3) -- (C33);	
	\draw[line width=\lww pt,blue] (intersection-2) -- (C12);			

	\path [name intersections={of={g3} and b0}];
	\draw[line width=\lww pt,blue,dotted] (intersection-3) -- (intersection-2);
	\draw[line width=\lww pt,blue] (intersection-3) -- (C13);	
	\draw[line width=\lww pt,blue] (intersection-2) -- (C23);			

	\foreach \i in {3,5,7}
	\draw[line width=\lww pt,co] (P1) -- (P\i);		
	
	\path [name intersections={of={g1} and g2}];		
	\draw[line width=\lww pt,co,dotted] (intersection-3) -- (intersection-2);
	\draw[line width=\lww pt,co,dotted] (intersection-1) -- (intersection-4);
	\draw[line width=\lww pt,co] (intersection-4) -- (D13);
	\draw[line width=\lww pt,co] (intersection-2) -- (D12);
	\draw[line width=\lww pt,co] (intersection-3) -- (P5);	

	\path [name intersections={of={g2} and g3}];		
	\draw[line width=\lww pt,co,dotted] (intersection-1) -- (intersection-2);
	\draw[line width=\lww pt,co,dotted] (intersection-3) -- (intersection-4);
	\draw[line width=\lww pt,co] (intersection-2) -- (D22);
	\draw[line width=\lww pt,co] (intersection-3) -- (D23);	
	\draw[line width=\lww pt,co] (intersection-1) -- (P3);	
	\draw[line width=\lww pt,co] (intersection-4) -- (P7);	

	\path [name intersections={of={g3} and g1}];		
	\draw[line width=\lww pt,co,dotted] (intersection-1) -- (intersection-2);
	\draw[line width=\lww pt,co,dotted] (intersection-3) -- (intersection-4);
	\draw[line width=\lww pt,co] (intersection-2) -- (D32);
	\draw[line width=\lww pt,co] (intersection-4) -- (D33);	
	\draw[line width=\lww pt,co] (intersection-1) -- (P5);	
	\draw[line width=\lww pt,co] (intersection-3) -- (P3);	
}

\foreach \i in {1,2,3,5,7}
  \draw[fill] (P\i) circle (\s pt);

\end{tikzpicture}

%% file: figures/n-st9-pp.tex
\tdplotsetmaincoords{0}{0}
\begin{tikzpicture}[tdplot_main_coords,scale=4,blend group=normal,tdplot_main_coords]
\def \s {.35};

\coordinate (P1) at (0, 0, 0.75);
\coordinate (P2) at  (0, 0.58, 0);
\coordinate (P3) at (-0.03, -0.11, 1);
\coordinate (P4) at (0, 0, 0.25);
\coordinate (P5) at (0.5, -0.29, 0);
\coordinate (P6) at (0.11, 0.03, 1);
\coordinate (P7) at (-0.23, -0.05, 1.44);
\coordinate (P8) at (-0.5, -0.29, 0);
\coordinate (P9) at (-0.08, 0.08, 1);

\coordinate (l1) at (-0.42568, -0.03877, 1.41077);
\coordinate (l2) at (0.47318, -0.10018, 1.06408);

\coordinate (P) at (-0.23259, -0.05196, 1.33629);

\draw[plum,fill=plum,opacity=.5,semitransparent] (P2) -- (P5) -- (P8) -- (P2);

\draw[blue,fill=blue,opacity=.5] (P4) -- (P5) -- (P6) -- (P4);
\draw[blue,fill=blue,opacity=.5] (P4) -- (P3) -- (P8) -- (P4);  
\draw[blue,fill=blue,opacity=.5] (P4) -- (P2) -- (P9) -- (P4);  

\draw[darkgreen,fill=darkgreen,opacity=.9] (P1) -- (P6) -- (P8) -- (P1);
\draw[darkgreen,fill=darkgreen,opacity=.9] (P1) -- (P2) -- (P3) -- (P1);  
\draw[darkgreen,fill=darkgreen,opacity=.9] (P1) -- (P5) -- (P9) -- (P1);

\draw[lgray,opacity=.8,line width=1.5pt] (P2) -- (P6);
\draw[lgray,opacity=.8,line width=1.5pt] (P3) -- (P5);
\draw[lgray,opacity=.8,line width=1.5pt] (P8) -- (P9);
\draw[red,opacity=.8,line width=1.5pt] (P1) -- (P4);

\draw[opacity=.8,line width=.5pt] (l1) -- (l2) node[right] {$\ell$};

\draw[fill] (P1) circle (\s pt);

\draw[plum,fill=plum,opacity=.7] (P3) -- (P6) -- (P9) -- (P3);

\foreach \i in {2,3,5,6,8,9}
  \draw[fill] (P\i) circle (\s pt);

\node[label=below:{$1$},inner sep=1pt,yshift=.7mm] at (P1) {};
\node[label=above:{$2$},inner sep=1pt] at (P2) {};
\node[label=below:{$3$},inner sep=.8pt] at (P3) {};
\node[label=above:{$4$},inner sep=1pt,yshift=-.7mm] at (P4) {};
\node[label=right:{$5$},inner sep=1pt] at (P5) {};
\node[label=right:{$6$},inner sep=1pt] at (P6) {};
\node[label=left:{$8$},inner sep=1pt] at (P8) {};
\node[label=left:{$9$},inner sep=1pt] at (P9) {};

\node[circle,draw,opacity=.8,minimum size=1mm,inner sep=1pt] at (P) {};
\node[label=above:{$P$},yshift=-1.5mm] at (P) {};
\end{tikzpicture}

%% file: figures/n-st9-ps.tex
\tdplotsetmaincoords{70}{-80}
\begin{tikzpicture}[tdplot_main_coords,scale=4,
blend group=normal,tdplot_main_coords]
\def \s {.35};

\coordinate (P1) at (0, 0, 0.75);
\coordinate (P2) at  (0, 0.58, 0);
\coordinate (P3) at (-0.03, -0.11, 1);
\coordinate (P4) at (0, 0, 0.25);
\coordinate (P5) at (0.5, -0.29, 0);
\coordinate (P6) at (0.11, 0.03, 1);
\coordinate (P7) at (-0.23, -0.05, 1.44);
\coordinate (P8) at (-0.5, -0.29, 0);
\coordinate (P9) at (-0.08, 0.08, 1);

\coordinate (l1) at (-1.13623, 0.00977, 1.68482);
\coordinate (l2) at (1.46373, -0.16786, 0.68203);

\coordinate (P) at (-0.23259, -0.05196, 1.33629);

\draw[lgray,opacity=.8,line width=1.5pt] (P2) -- (P6);
  
\draw[plum,fill=plum,opacity=.7,semitransparent] (P2) -- (P5) -- (P8) -- (P2);

\draw[blue,fill=blue,opacity=.35,semitransparent] (P4) -- (P5) -- (P6) -- (P4);

\draw[darkgreen!90!black,fill=darkgreen!100!black,opacity=.55,semitransparent] (P1) -- (P9) -- (P5) -- (P1); 

\draw[blue,fill=blue,opacity=.5] (P4) -- (P3) -- (P8) -- (P4);   

\draw[darkgreen!85!black,fill=darkgreen!88!black,opacity=.5] (P1) -- (P6) -- (P8) -- (P1);

\draw[darkgreen!60!black,fill=darkgreen!70!black,opacity=1] (P1) -- (P2) -- (P3) -- (P1);

\draw[blue,fill=blue,opacity=.7] (P4) -- (P2) -- (P9) -- (P4);
 
\draw[plum,fill=plum,opacity=.7] (P3) -- (P6) -- (P9) -- (P3);

\draw[red,opacity=.5,line width=1.5pt] (P1) -- (P4);

\draw[lgray,opacity=1,line width=1.5pt] (P3) -- (P5);
\draw[lgray,opacity=1,line width=1.5pt] (P8) -- (P9);

\draw[opacity=.8,line width=.5pt] (l1) -- (l2) node[right] {$\ell$};

\foreach \i in {1,2,3,4,5,6,8,9}
  \draw[fill] (P\i) circle (\s pt);

\node[label=above:{$1$},inner sep=1pt,yshift=-.5mm] at (P1) {};
\node[label=left:{$2$},inner sep=1pt] at (P2) {};
\node[label=right:{$3$},inner sep=.8pt] at (P3) {};
\node[label=below:{$4$},inner sep=1pt,yshift=0mm,xshift=-.5mm] at (P4) {};
\node[label=right:{$5$},inner sep=1pt] at (P5) {};
\node[label=right:{$6$},inner sep=1pt] at (P6) {};
\node[label=right:{$8$},inner sep=1pt] at (P8) {};
\node[label=left:{$9$},inner sep=1pt] at (P9) {};

\node[circle,draw,opacity=.8,minimum size=1mm,inner sep=1pt] at (P) {};
\node[label=above:{$P$},yshift=-1.5mm] at (P) {};

\end{tikzpicture}

%% file: figures/n-st9-cp.tex
\tdplotsetmaincoords{0}{0}
\begin{tikzpicture}[tdplot_main_coords,scale=5,blend group=normal,tdplot_main_coords]
\def \s {.3};

\coordinate (P1) at (0, 0, 0.75);
\coordinate (P2) at  (0, 0.58, 0);
\coordinate (P3) at (-0.03, -0.11, 1);
\coordinate (P4) at (0, 0, 0.25);
\coordinate (P5) at (0.5, -0.29, 0);
\coordinate (P6) at (0.11, 0.03, 1);
\coordinate (P7) at (-0.23, -0.05, 1.44);
\coordinate (P8) at (-0.5, -0.29, 0);
\coordinate (P9) at (-0.08, 0.08, 1);

\draw[plum,fill=plum,opacity=.5,semitransparent] (P2) -- (P5) -- (P8) -- (P2);

\draw[blue,fill=blue,opacity=.5] (P4) -- (P5) -- (P6) -- (P4);
\draw[blue,fill=blue,opacity=.5] (P4) -- (P3) -- (P8) -- (P4);  
\draw[blue,fill=blue,opacity=.5] (P4) -- (P2) -- (P9) -- (P4);  

\draw[darkgreen,fill=darkgreen,opacity=.9] (P1) -- (P6) -- (P8) -- (P1);
\draw[darkgreen,fill=darkgreen,opacity=.9] (P1) -- (P2) -- (P3) -- (P1);  
\draw[darkgreen,fill=darkgreen,opacity=.9] (P1) -- (P5) -- (P9) -- (P1);

\draw[red,fill=red,opacity=.5,line width=1.5pt] (P4) -- (P1) -- (P7) -- cycle;

\draw[plum,fill=plum,opacity=.85] (P3) -- (P6) -- (P9) -- (P3);

\filldraw[lgray,opacity=1,line width=.5pt] (P7) -- (P3) -- (P5) -- (P7);
\filldraw[lgray,opacity=1,line width=.3pt] (P7) -- (P8) -- (P9) -- (P7);
\filldraw[lgray,opacity=.65,line width=.2pt] (P7) -- (P2) -- (P6) -- (P7);

\draw[fill,opacity=.3] (P1) circle (\s pt);

\foreach \i in {3,6,9}
  \draw[fill,opacity=.6] (P\i) circle (\s pt);

\foreach \i in {2,5,8,7}
  \draw[fill] (P\i) circle (\s pt);

\node[label=below:{$1$},inner sep=1pt,yshift=.7mm] at (P1) {};
\node[label=above:{$2$},inner sep=1pt] at (P2) {};
\node[label=below:{$3$},inner sep=.8pt] at (P3) {};
\node[label=above:{$4$},inner sep=1pt,yshift=-.7mm] at (P4) {};
\node[label=right:{$5$},inner sep=1pt] at (P5) {};
\node[label=right:{$6$},inner sep=1pt] at (P6) {};
\node[label=left:{$8$},inner sep=1pt] at (P8) {};
\node[label=left:{$9$},inner sep=1pt] at (P9) {};
\node[label=above left:{$7$},inner sep=0pt,xshift=0.5mm,yshift=-2mm] at (P7) {};

\end{tikzpicture}

%% file: figures/n-st9-cs.tex
\tdplotsetmaincoords{24}{-31}
\begin{tikzpicture}[tdplot_main_coords,scale=5,
blend group=normal,tdplot_main_coords]
\def \s {.28};

\coordinate (P1) at (0, 0, 0.75);
\coordinate (P2) at  (0, 0.58, 0);
\coordinate (P3) at (-0.03, -0.11, 1);
\coordinate (P4) at (0, 0, 0.25);
\coordinate (P5) at (0.5, -0.29, 0);
\coordinate (P6) at (0.11, 0.03, 1);
\coordinate (P7) at (-0.23, -0.05, 1.44);
\coordinate (P8) at (-0.5, -0.29, 0);
\coordinate (P9) at (-0.08, 0.08, 1);

\coordinate (P) at (-0.23259, -0.05196, 1.33629);

\draw[plum,fill=plum,opacity=.7,semitransparent] (P2) -- (P5) -- (P8) -- (P2);

\draw[blue,fill=blue,opacity=.35,semitransparent] (P4) -- (P2) -- (P9) -- (P4);
\draw[blue,fill=blue,opacity=.35,semitransparent] (P4) -- (P5) -- (P6) -- (P4);

\draw[draw=darkgreen!80!black,fill=darkgreen,opacity=.8] (P1) -- (P9) -- (P5) -- (P1); 
\draw[draw=darkgreen!70!black,fill=darkgreen,opacity=1] (P1) -- (P6) -- (P8) -- (P1);
\draw[draw=darkgreen!60!black,fill=darkgreen,opacity=1] (P1) -- (P2) -- (P3) -- (P1);

\draw[plum,fill=plum!90!black,opacity=.75] (P3) -- (P6) -- (P9) -- (P3);

\filldraw[draw=lgray!60!black,fill=lgray!90!black,opacity=.8,line width=.3pt] (P7) -- (P2) -- (P6) -- (P7);
\filldraw[draw=lgray!60!black,fill=lgray!90!black,opacity=.7,line width=.3pt] (P7) -- (P8) -- (P9) -- (P7);
\draw[red,fill=red,opacity=.6] (P4) -- (P1) -- (P7) -- cycle;

\draw[blue,fill=blue,opacity=.7] (P4) -- (P3) -- (P8) -- (P4);

\filldraw[draw=lgray!60!black,fill=lgray!85!black,opacity=1,line width=.3pt] (P7) -- (P3) -- (P5) -- (P7);

\foreach \i in {1,2,3,4,5,6,8,9,7}
  \draw[fill] (P\i) circle (\s pt);

\node[label=left:{$1$},inner sep=1pt,xshift=.5mm] at (P1) {};
\node[label=left:{$2$},inner sep=1pt] at (P2) {};
\node[label=right:{$3$},inner sep=.8pt] at (P3) {};
\node[label=below:{$4$},inner sep=1pt,yshift=0mm] at (P4) {};
\node[label=right:{$5$},inner sep=1pt] at (P5) {};
\node[label=right:{$6$},inner sep=1pt] at (P6) {};
\node[label=below:{$8$},inner sep=1pt] at (P8) {};
\node[label=right:{$9$},inner sep=1pt] at (P9) {};

\node[label=above:{$7$},inner sep=0pt,yshift=-.5mm] at (P7) {};

\end{tikzpicture}